\documentclass[11pt,a4paper]{article}
\usepackage[utf8]{inputenc}
\usepackage{authblk}
\usepackage[top=25truemm,bottom=25truemm,left=25truemm,right=25truemm,footskip=10truemm]{geometry}
\usepackage{booktabs}
\usepackage{amsmath}
\usepackage{amsthm}
\usepackage{float}
\usepackage{amssymb}
\usepackage{amsfonts}
\usepackage{ascmac}
\usepackage{scalefnt}
\usepackage{here}
\usepackage{indentfirst}
\usepackage{algorithmic}
\usepackage{algorithm}
\usepackage{setspace}
\usepackage{color}
\usepackage[pdftex]{graphicx}
\usepackage{subcaption}
\usepackage{lscape}
\usepackage{natbib}
\usepackage{url}
\usepackage{placeins}
\usepackage{tikz}
\usepackage{multirow}
\DeclareMathOperator*{\argmin}{arg\,min}

\newtheorem{theorem}{Theorem}[section]
\newtheorem{proposition}[theorem]{Proposition}
\newtheorem{corollary}[theorem]{Corollary}
\newtheorem{lemma}[theorem]{Lemma}
\newtheorem{assumption}{Assumption}[section]
\newtheorem{definition}{Definition}[section]

\title{
Zero-Inflated Logistic Regression Models with Shared Design: Identifiability, Existence of Estimates, and a Relabeling Rule
}

\author[1,*]{Yui Tomo}
\author[2]{Shinto Eguchi}
\author[1]{Daisuke Yoneoka}
\affil[1]{Department of Epidemiology, National Institute of Infectious Diseases, Japan Institute for Health Security, 1-23-1 Toyama, Shinjuku-ku, Tokyo 162-0052, Japan}
\affil[2]{The Institute of Statistical Mathematics, 10-3 Midori-cho, Tachikawa, Tokyo 190-8562, Japan}
\affil[*]{E-mail: tomo.y@jihs.go.jp}
\date{}

\begin{document}

\maketitle

\begin{abstract}
The zero-inflated logistic regression model accommodates binary responses with excess zeros, which often arise from a latent mixture of susceptible and insusceptible subpopulations or asymmetric misclassification of the response.
The model has two components: regression for the binary response and a latent binary indicator for the zero-inflation state.
In applied settings, it is common to use the same design matrix for both components if there is no prior knowledge.
However, this shared-design specification lacks guaranteed identifiability of the regression parameters, as established in prior works.
This paper investigates the theoretical properties of the zero-inflated logistic regression model under the shared-design setting and computational methods for applications.
First, to motivate the use of the zero-inflated model, we prove that ignoring the zero-inflation mechanism can lead to a sign flip in the pseudo-true coefficient value relative to the true value.
We then establish sufficient conditions for the existence of the maximum likelihood estimate.
As a main result, we establish that the model under the shared-design setting is identifiable up to exchange symmetry of the parameters for two components and that the expected log-likelihood has a unique maximizer on the resulting quotient space.
The posterior bimodality is examined using a P\'olya-Gamma Gibbs sampler with replica exchange.
Finally, we propose a simple relabeling rule to select a single ordered parameter pair, and evaluate its performance through simulation studies and an application to self-reported diabetes data.
\end{abstract}


\section{Introduction}

The logistic regression model is a standard tool for binary outcomes and remains attractive due to its simplicity and interpretability.
However, in many applied settings, the observed response contains more zeros than a standard logistic model can accommodate.
For example, such excess zeros may arise from a latent mixture of susceptible and insusceptible subpopulations, such as biological immunity, or one-sided outcome misclassification, such as the failure to record an event due to delayed reporting.
To accommodate these situations, a natural extension is a zero-inflated logistic regression model \citep{hall2000zero,diop2011maximum}.
This model expresses the observed response as a mixture of a standard logistic regression and a structural zero, where the latent binary indicator for the zero-inflation state is itself modeled via a logistic regression.
We refer to these two regression components as the ordinary logistic regression component and the structural-zero component, respectively.
This formulation captures complex data-generating mechanisms while retaining interpretability.

Earlier work includes the three-parameter logistic model in psychometrics and related methods in ecology and epidemiology \citep{wainer2007testlet,komori2016asymmetric,nagelkerke2015estimating}.
Furthermore, robust estimation approaches based on label-noise modeling or robust divergence have also been studied \citep{bootkrajang2012label,bootkrajang2013classification,hung2018robust,fujisawa2008robust}.
From the perspective of excess zeros, these settings are closely connected: positive outcomes that are systematically recorded as zeros induce a structural-zero mechanism in the observed data.
From the theoretical perspective, \cite{diop2011maximum} established sufficient conditions for identifiability of the zero-inflated logistic regression model.
One such condition requires at least one continuous covariate that appears in one component but not in the other.

Although the zero-inflated logistic model has been explored practically and theoretically, there remain several gaps in the understanding of its theoretical properties.
In particular, this study focuses on the practically important scenario where no reliable prior information is available regarding which covariates should enter both components.
A common empirical choice is then to use the same design matrix in both components of the model.
In this case, the model faces a theoretical difficulty: once the two components share the same covariates, the likelihood becomes invariant under exchange of the two coefficient vectors.
The model is therefore not identifiable as an ordered pair, and both optimization and posterior simulation may exhibit symmetric multiple modes.
This symmetry is analogous to the label-switching phenomenon in mixture models \citep{fruhwirth2006finite}.
Specifically, once the same covariates are used in both components, exchanging the two coefficient vectors does not change the likelihood values.
However, such non-identifiability has yet to be characterized.

Another theoretical issue concerns the existence of the maximum likelihood estimate.
For ordinary logistic regression, non-existence of the estimate under separation is well-known \citep{Albert1984,Silvapulle1981}.
For the zero-inflated logistic model, however, the analogous conditions have received less attention.

Our contributions are fourfold.
First, we prove that model misspecification can lead to a sign flip in the pseudo-true parameter relative to the true coefficient.
Second, we introduce a double-separation condition for the zero-inflated logistic model and derive sufficient conditions for the existence of the estimates.
Third, we establish that the model under a shared design matrix is identifiable up to exchange symmetry and that, under mild regularity conditions, the expected log-likelihood has a unique maximizer on the resulting quotient space.
Furthermore, we investigate the resulting multimodality numerically through a P\'olya-Gamma Gibbs sampler integrated with replica exchange \citep{polson2013bayesian, swendsen1986replica}.
Fourth, we propose a relabeling rule based on an estimate from the ordinary logistic regression model and conduct a numerical study.

The remainder of the paper is structured as follows.
Section~\ref{sec:zilr_model} introduces the model, studies the sign-flip phenomenon under misspecification, and presents sufficient conditions for the existence and non-existence of the estimates.
Section~\ref{sec:model_shared_design} studies identifiability under the shared design setting, formalizes the exchange-symmetry structure, and provides a numerical illustration of the resulting posterior bimodality.
Section~\ref{sec:relabeling} proposes a relabeling rule and conducts a simulation study.
Section~\ref{sec:application} presents an illustrative application to NHANES self-reported diabetes data.
Section~\ref{sec:discussion} discusses limitations and future work. Section~\ref{sec:conclusion} concludes the paper.

\section{Zero-inflated Logistic Regression Model}
\label{sec:zilr_model}

In this section, we define the zero-inflated logistic regression model and study two basic aspects of the model before turning to the shared-design case.
First, we examine the consequence of misspecification when zero inflation is ignored and a standard logistic regression model is fitted instead.
Second, we consider the existence of maximum likelihood estimates and introduce separation-type conditions that provide sufficient criteria for non-existence and existence.

\subsection{Model Definition}

Let $y_{i} \in \{0,1\}$ ($i = 1,\ldots,n$) denote the binary responses for $n$ observations, and let $x_i=\left(1,x_{i,1}, \ldots, x_{i, d-1}\right)^{\top} \in \mathbb{R}^d$ denote the covariate vector.
Let $\beta=\left(\beta_0, \ldots, \beta_{d-1}\right)^{\top} \in \mathbb{R}^d$ be the corresponding parameter vector.
We assume that a latent variable $h_i \in \{0,1\}$ determines whether the response is constrained to be zero, where $h_i = 0$ indicates the zero state.
Let $z_i=\left(1,z_{i,1}, \ldots, z_{i,p-1}\right)^{\top} \in \mathbb{R}^p$ denote the covariate vector for $h_i$, and let $\gamma=\left(\gamma_0, \ldots, \gamma_{p-1}\right)^{\top} \in \mathbb{R}^p$ denote the corresponding parameter vector.
Let $F(\cdot)$ be the inverse logit function: 
\begin{align*}
    F(\mu) := {\exp(\mu)}/{\left\{1+\exp(\mu)\right\}}
    ,
    \quad
    \text{for}
    \quad
    \mu \in \mathbb{R}
    .
\end{align*}

Then, the zero-inflated logistic regression model is defined as
\begin{align*}
    p(y_i \mid x_i, z_i, \beta, \gamma) 
    =
    q(h_i = 1 \mid z_i, \gamma) \cdot p(y_i \mid x_i, \beta) + q(h_i = 0 \mid z_i, \gamma) \cdot \mathrm{I}(y_i = 0)
    ,
\end{align*}
where $p(y_i = 1 \mid x_i, \beta)=F(\beta^{\top} x_i)$ and $q(h_i = 1 \mid z_i, \gamma)=F(\gamma^{\top} z_i)$.
We refer to the logistic regression components $p(y_i \mid x_i, \beta) = F(\beta^{\top} x_i)^{y_i} (1-F(\beta^{\top} x_i))^{1-y_i}$ and $q(h_i \mid z_i, \gamma) = F(\gamma^{\top} z_i)^{h_i} (1-F(\gamma^{\top} z_i))^{1-h_i}$ as the ordinary logistic regression component and the structural-zero component, respectively.
In this context, $q(h_i = 1 \mid z_i, \gamma)$ represents the probability that the $i$-th observation is not a structural zero, corresponding to the probability of susceptibility to the event.
Then, we have
\begin{align*}
    p(y_i \mid x_i, z_i, \beta, \gamma)
    = 
    \left\{F(\gamma^{\top} z_i) F(\beta^{\top} x_i) \right\}^{y_i}
    \left\{1-F(\gamma^{\top} z_i) F(\beta^{\top} x_i)\right\}^{1 - y_i}
    .
\end{align*}
Therefore, the log-likelihood function is defined as
\begin{align*}
    L(\beta, \gamma)
    =
    \sum_{i=1}^n\left[y_i \log \left\{F(\gamma^{\top} z_i) F(\beta^{\top} x_i)\right\}+(1-y_i) \log \left\{1-F(\gamma^{\top} z_i) F(\beta^{\top} x_i)\right\}\right]
    .
\end{align*}

\subsection{Sign-Flip Phenomenon Under Misspecification}

\label{subsec:misspec}

To motivate the use of the zero-inflated logistic regression model, we examine the consequences of ignoring zero-inflation and fitting a standard logistic regression model to data generated from the zero-inflated model.
When the standard logistic regression model is applied to data with excess zeros, the resulting estimator can be severely biased.
In particular, when a covariate is positively associated with the response $y_i$ but negatively associated with the latent indicator $h_i$ (or vice versa), the fitted logistic regression model may estimate a regression coefficient whose sign is opposite to the underlying true value.
In this subsection, we formalize this phenomenon under a random-design setting.

Let $x \in \mathbb{R}^{d}$ and $z \in \mathbb{R}^{p}$ denote covariate vectors in the ordinary logistic regression component and the structural-zero component, respectively, and suppose that they share a common covariate indexed by $j$:
\begin{align*}
    x
    &= (1,\tilde{x}_j,\tilde{x}_{-j}^{\top})^{\top}
    ,
    \\
    z
    &= (1,\tilde{z}_j,\tilde{z}_{-j}^{\top})^{\top}
    ,
\end{align*}
where $\tilde{x}_{-j} \in \mathbb{R}^{d-2}$, $\tilde{z}_{-j}\in \mathbb{R}^{p-2}$, and the $j$th covariate is shared between the two vectors: $\tilde{x}_j = \tilde{z}_j$.
The corresponding regression coefficients are decomposed as
\begin{align*}
    \beta
    &=
    (\beta_0,\beta_j,\beta_{-j}^{\top})^{\top}
    ,
    \\
    \gamma
    &=
    (\gamma_0,\gamma_j,\gamma_{-j}^{\top})^{\top}
    .
\end{align*}
Without loss of generality, we focus on a covariate that has a positive effect on $y$ but a negative effect on $h$, and let
\begin{align*}
    \beta_j
    &=:
    a > 0,
    \\
    \gamma_j
    &=:
    c < 0
    .
\end{align*}
Under the zero-inflated logistic regression model, the conditional event occurrence probability can be written as
\begin{align*}
    \pi_c(x,z)
    :=
    F(\beta_0+a \tilde{x}_j+\beta_{-j}^{\top} \tilde{x}_{-j})
    F(\gamma_0+c \tilde{x}_j+\gamma_{-j}^{\top} \tilde{z}_{-j})
    .
\end{align*}

We then consider the ordinary logistic regression model where the conditional event occurrence probability is specified as $F(\theta_0 + t \tilde{x}_j + \theta_{-j}^{\top} \tilde{x}_{-j})$.
Here, $(\theta_0,t,\theta_{-j})\in\mathbb{R} \times \mathbb{R} \times \mathbb{R}^{d-2}$ is the regression coefficient vector.
Under the true zero-inflated model, the expected log-likelihood function of this misspecified logistic regression is given by
\begin{align*}
    &\mathcal{L}_c(\theta_0,t,\theta_{-j})
    \\
    &\quad
    :=
    \mathbb{E}_{(x,z)}\left[
    \pi_c(x,z)\,\log F(\theta_0 + t \tilde{x}_j + \theta_{-j}^{\top} \tilde{x}_{-j})
    + \{1-\pi_c(x,z)\}\,\log\{1-F(\theta_0 + t \tilde{x}_j + \theta_{-j}^{\top} \tilde{x}_{-j})\}
    \right]
    ,
\end{align*}
where the expectation is taken with respect to the joint distribution of $(x,z)$.

We then define the profile objective function
\begin{align*}
    g_c(t)
    &:=
    \sup_{(\theta_0,\theta_{-j}) \in \mathbb{R} \times \mathbb{R}^{d-2}}
      \mathcal{L}_c(\theta_0,t,\theta_{-j})
    ,
\end{align*}
and let $t^{*}(c) \in \arg\max_{t \in \mathbb{R}} g_c(t)$.
The following theorem shows that, when the magnitude of the negative association in the zero-inflation component is sufficiently large, every pseudo-true value of the $j$th coefficient is negative.

\begin{theorem}
\label{thm:sign-flip}
Suppose that Assumption~\ref{ass:basic} in
Appendix~\ref{apx:proof_misspec} holds.
Then there exists a constant $C_0<0$ such that
\begin{align*}
    c \leq C_0
    \quad
    \Longrightarrow
    \quad
    \arg\max_{t \in \mathbb{R}} g_c(t) \subset (-\infty,0)
    .
\end{align*}
\end{theorem}

See Appendix~\ref{apx:proof_misspec} for the proof.
This theorem implies that, although the true coefficient satisfies $\beta_j=a>0$, every pseudo-true value of the same coefficient under the misspecified conventional logistic regression model is negative when the $|\gamma_j| = |c|$ is sufficiently large in the opposite sign direction.
Specifically, any choice $t^{*}(c)\in \arg\max_{t \in \mathbb{R}} g_c(t)$ satisfies $t^{*}(c)<0$.

\subsection{Maximum Likelihood Estimation}

The maximum likelihood estimator (MLE) $(\hat{\beta}^{\top}, \hat{\gamma}^{\top})^{\top}  \in \mathbb{R}^{d + p}$ is defined as any maximizers of $L(\beta, \gamma)$.
However, due to the structure of the zero-inflated logistic model, the existence of an estimate is not always guaranteed.
To investigate this issue, we introduce the concept of double separation, an analogue of the separation condition for the standard logistic regression.

\begin{definition}[Double separation]
The dataset $\{(y_i,x_i,z_i)\}_{i=1}^n$ is said to satisfy double separation if there exist non-zero vectors 
${v}\in\mathbb{R}^{d}$ and ${w}\in\mathbb{R}^{p}$ such that for every $i=1,\dots,n$, 
\begin{align*}
    \begin{cases}
        v^{\top}x_i \;\geq 0,\quad w^{\top}z_i \;\geq 0 & \text{if } y_i=1
        ,
        \\
        v^{\top}x_i \;\leq 0,\quad w^{\top}z_i \;\leq 0 & \text{if } y_i=0
        ,
    \end{cases}
\end{align*}
and either of the inequalities is strict for at least one observation: there exists $j\in\{1,\dots,n\}$ such that either $y_{j}=1$ with $(v^{\top}x_{j},w^{\top}z_{j})\neq(0,0)$ or $y_{j}=0$ with $(v^{\top}x_{j},w^{\top}z_{j})\neq(0,0)$.
\end{definition}

\begin{proposition}
\label{prop:MLEnotexist}
If the dataset satisfies double separation, then the log-likelihood $L(\beta,\gamma)$ has no maximizer in $\mathbb{R}^{d}\times\mathbb{R}^{p}$.
\end{proposition}

See Appendix~\ref{apx:proof_mle_notexist} for the proof.

Furthermore, we introduce the following condition to guarantee the existence of a maximum likelihood estimate.

\begin{definition}[$\varepsilon$--Double--non-separation]
The dataset $\{(y_i,x_i,z_i)\}_{i=1}^{n}$ is said to satisfy {$\varepsilon$--double--non-separation} if there exists a constant $\varepsilon>0$ such that 
\begin{align*}
    \inf_{(v,w):\|v\|^{2}+\|w\|^{2}=1}
    \max\left\{
    -\min_{i\in \{i:y_i=1\}}\{v^{\top}x_i+w^{\top}z_i\},
    ~
    \max_{i\in \{i:y_i=0\}}\min\{v^{\top}x_i,~w^{\top}z_i\}
    \right\}
    \geq \varepsilon
    .
\end{align*}
\end{definition}

We then establish the following proposition, which provides a sufficient condition for the existence of an estimate.

\begin{proposition}
\label{prop:MLEexist}
If the dataset satisfies $\varepsilon$--double--non-separation for some $\varepsilon>0$, then a maximizer of $L(\beta,\gamma)$ exists.
\end{proposition}

See Appendix~\ref{apx:proof_mle_exist} for the proof. 

The $\varepsilon$--double--non-separation condition implies that for any unit direction $(v,w)$, there exists at least one observation located at a signed distance of at least $\varepsilon$ from the separating hyperplane associated with that direction.
In other words, the dataset maintains a uniform margin of width $\varepsilon$ that prevents double separation.

Even when $\varepsilon$--double--non-separation holds with $\varepsilon$ close to zero, the surface of the log-likelihood function may be nearly flat along certain directions, and numerical optimization may be unstable in practice.
In such settings, penalized estimation approaches may be useful.

A closely related phenomenon in this section arises in ordinary logistic regression, where data separation leads to the non-existence of the maximum likelihood estimate \citep{Albert1984,Silvapulle1981}.

\section{Shared-Design Model}
\label{sec:model_shared_design}

We now focus on the case $p=d$ and $z_i=x_i$ for all $i=1,\ldots,n$.
This is the configuration that arises when the analyst has no prior information with which to distinguish covariates for the ordinary logistic regression component $x_i$ from covariates for the structural-zero component $z_i$ and therefore uses the same design matrix in both components of the model.
In that case, the conditional probability of $y_i$ is expressed as
\begin{align*}
    p(y_i\mid x_i,\beta,\gamma)
    =
    \left\{F(\gamma^\top x_i)F(\beta^\top x_i)\right\}^{y_i}
    \left\{1-F(\gamma^\top x_i)F(\beta^\top x_i)\right\}^{1-y_i}
    ,
\end{align*}
for $i=1,\ldots,n$.
We refer to this model as a shared-design model.
The log-likelihood function is
\begin{align*}
    L(\beta, \gamma)
    =
    \sum_{i=1}^n\left[y_i \log \left\{F(\gamma^{\top} x_i) F(\beta^{\top} x_i)\right\}+(1-y_i) \log \left\{1-F(\gamma^{\top} x_i) F(\beta^{\top} x_i)\right\}\right]
    ,
\end{align*}
which satisfies the exchange symmetry
\begin{align*}
    L(\beta,\gamma)
    =
    L(\gamma,\beta)
    .
\end{align*}
Motivated by this symmetry, we define the equivalence relation
\begin{align*}
    (\beta_1,\gamma_1)\sim(\beta_2,\gamma_2)
    \quad\Longleftrightarrow\quad
    (\beta_1,\gamma_1)=(\beta_2,\gamma_2)
    \quad \text{or}\quad
    (\beta_1,\gamma_1)=(\gamma_2,\beta_2),
\end{align*}
and define $[\beta,\gamma]$ for the corresponding equivalence class.

\subsection{Prior Work on Identifiability}

\cite{diop2011maximum} established sufficient conditions for identifiability of the zero-inflated logistic regression model.
A key condition in their argument is the availability of a continuous covariate that appears in one component but not the other.
When the same covariates are used in both components, this source of asymmetry disappears.
Therefore, the shared-design setting is a boundary case in which the ordinary guarantee of identifiability fails.

\subsection{Identifiability under Shared Design}

We formulate identifiability of the shared design model.
Let $x=(1,\tilde{x}_{-0}^\top)^\top$ where $\tilde{x}_{-0} \in \mathbb{R}^{d-1}$.
We first consider the following support condition.

\begin{itemize}
    \item[(C1)] The support of $\tilde{x}_{-0}$ contains a nonempty open subset $\mathcal{U}\subset\mathbb{R}^{d-1}$.
\end{itemize}

Under condition {(C1)}, we obtain the following basic identifiability result.

\begin{proposition}
\label{prop:two_mode_c1}
Suppose that condition {(C1)} holds.
Let $\beta_1=(\beta_{1,0},\beta_{1,-0}^\top)^\top$ and $\gamma_1=(\gamma_{1,0},\gamma_{1,-0}^\top)^\top$, and suppose that $\beta_{1,-0}$, $\gamma_{1,-0}$, $\beta_{1,-0}+\gamma_{1,-0}$ are pairwise distinct.
Suppose that for two parameter pairs $(\beta_1,\gamma_1)$ and $(\beta_2,\gamma_2)$,
\begin{align}
    \label{eq:equiv_prob}
    p\left(y \mid x, \beta_1, \gamma_1\right) = p\left(y \mid x, \beta_2, \gamma_2\right)
    ,
\end{align}
holds for all $y \in \{0,1\}$ and all $x=(1,\tilde{x}_{-0}^\top)^\top$ for $\tilde{x}_{-0} \in \mathcal{U}$.
Then, we have either $(\beta_1, \gamma_1) = (\beta_2, \gamma_2)$ or $(\beta_1, \gamma_1) = (\gamma_2, \beta_2)$.
\end{proposition}

We next extend Proposition~\ref{prop:two_mode_c1} to a mixed support for continuous and discrete covariates.
Let $\tilde{x}_{-0}=(\tilde{x}_{-0}^{(1)\top},\tilde{x}_{-0}^{(2)\top})^\top$, where $\tilde{x}_{-0}^{(1)} \in \mathbb{R}^{r}$ and $\tilde{x}_{-0}^{(2)} \in \mathbb{R}^{s}$ with $r \geq 1$, $s \geq 0$, and $r+s=d-1$.
Here, $\tilde{x}_{-0}^{(1)}$ denotes the subvector of continuously distributed covariates, and $\tilde{x}_{-0}^{(2)}$ denotes the subvector of discretely distributed covariates.
We consider the following support condition.

\begin{itemize}
    \item[(C2)] There exists a subset $\mathcal{S} \subset \operatorname{supp}(\tilde{x}_{-0}^{(2)})$ such that $\operatorname{aff}(\mathcal{S}) = \mathbb{R}^{s}$, and, for each $\xi \in \mathcal{S}$, the conditional support of $\tilde{x}_{-0}^{(1)}$ given $\tilde{x}_{-0}^{(2)} = \xi$ contains a nonempty open subset $\mathcal{U}_{\xi} \subset \mathbb{R}^{r}$.
\end{itemize}

When $s=0$, condition {(C2)} reduces to condition {(C1)}.
Let $\beta_{j,-0}=(\beta_{j,-0}^{(1)\top},\beta_{j,-0}^{(2)\top})^\top$ and $\gamma_{j,-0}=(\gamma_{j,-0}^{(1)\top},\gamma_{j,-0}^{(2)\top})^\top$, corresponding to the decomposition of $\tilde{x}_{-0}$ into $\tilde{x}_{-0}^{(j)}$ for $j=1,2$.
Under condition {(C2)}, we obtain the following extended result.

\begin{theorem}
\label{thm:two_mode}
Suppose that condition {(C2)} holds.
Let $\beta_1=(\beta_{1,0},\beta_{1,-0}^\top)^\top$ and $\gamma_1=(\gamma_{1,0},\gamma_{1,-0}^\top)^\top$, and suppose that $\beta_{1,-0}^{(1)}$, $\gamma_{1,-0}^{(1)}$, $\beta_{1,-0}^{(1)}+\gamma_{1,-0}^{(1)}$ are pairwise distinct.
Suppose that for two parameter pairs $(\beta_1,\gamma_1)$ and $(\beta_2,\gamma_2)$, \eqref{eq:equiv_prob} holds for all $y \in \{0,1\}$ and all $x=(1,\tilde{x}_{-0}^{(1)\top},\xi^\top)^\top$ for $\xi \in \mathcal{S}$ and $\tilde{x}_{-0}^{(1)} \in \mathcal{U}_{\xi}$.
Then, we have either $(\beta_1, \gamma_1) = (\beta_2, \gamma_2)$ or $(\beta_1, \gamma_1) = (\gamma_2, \beta_2)$.
\end{theorem}

We now establish the following identifiability result.

\begin{corollary}
\label{cor:quotient_mle}
Suppose that condition {(C2)} holds, and suppose that the shared-design model is correctly specified with true parameter pair $(\beta^\ast,~ \gamma^\ast)$.
Let $\beta^\ast=(\beta_0^\ast,~ \beta_{-0}^{\ast\top})^\top$ and $\gamma^\ast=(\gamma_0^\ast,~ \gamma_{-0}^{\ast\top})^\top$, where $\beta_{-0}^\ast=(\beta_{-0}^{\ast(1)\top},\beta_{-0}^{\ast(2)\top})^\top$ and $\gamma_{-0}^\ast=(\gamma_{-0}^{\ast(1)\top},\gamma_{-0}^{\ast(2)\top})^\top$, and suppose that $\beta_{-0}^{\ast(1)}$, $\gamma_{-0}^{\ast(1)}$, $\beta_{-0}^{\ast(1)}+\gamma_{-0}^{\ast(1)}$ are pairwise distinct.
Then the expected log-likelihood
\begin{align*}
    \mathcal{L}(\beta,~ \gamma)
    :=
    \mathbb{E}_{(x,y)}\left[
    y\log\left\{F(x^\top\beta)F(x^\top\gamma)\right\}
    +(1-y)\log\left\{1-F(x^\top\beta)F(x^\top\gamma)\right\}
    \right]
    ,
\end{align*}
is uniquely maximized on $(\mathbb{R}^d\times\mathbb{R}^d)/{\sim}$ at the class $[\beta^\ast,~ \gamma^\ast]$.
\end{corollary}

See Appendix~\ref{apx:proof_two_mode} for the proofs of Proposition~\ref{prop:two_mode_c1}, Theorem~\ref{thm:two_mode}, and Corollary~\ref{cor:quotient_mle}.
Proposition~\ref{prop:two_mode_c1} gives the basic identifiability result under a fully continuous support, while Theorem~\ref{thm:two_mode} extends it to a mixed support for continuous and discrete covariates.
Corollary~\ref{cor:quotient_mle} clarifies the inferential target under the shared-design setting.
At the population level, the equivalence class $[\beta, \gamma]$ is uniquely identified, yet the model cannot distinguish which component corresponds to the event occurrence process and which to the structural-zero process.
These results establish identifiability over the support of $x$ and do not imply the uniqueness of the likelihood maximizer in finite samples.
Consequently, resolving this label ambiguity requires either external information or some relabeling rule.

These results also relate the shared-design model to the identifiability of finite mixture models \citep{teicher1963identifiability, yakowitz1968identifiability}.
In this context, identifiability is formulated as the uniqueness of the finite-mixture representation.
Therefore, under an ordered-component parameterization, this corresponds to uniqueness up to permutation of component labels.
In the present model, the relevant symmetry is not a literal permutation of mixture components but the exchange symmetry of the pair of parameter vectors $(\beta, \gamma)$ for both components.
In this sense, $[\beta,\gamma]$ is the natural inferential target.

The condition that $\beta_{-0}^\ast$, $\gamma_{-0}^\ast$, and $\beta_{-0}^\ast+\gamma_{-0}^\ast$ are pairwise distinct, which is inherited from Theorem~\ref{thm:two_mode}, is a condition on the true parameter vector.
This is a generic condition and can be assumed to hold in practice.

\subsection{Numerical Confirmation of Bimodality}
\label{sec:numerical_bimodality}

To illustrate the exchange symmetry established in Theorem~\ref{thm:two_mode} and Corollary~\ref{cor:quotient_mle}, we examined the posterior distribution of $(\beta,\gamma)$ under the shared-design setting via a Markov chain Monte Carlo (MCMC) algorithm.
We set the number of non-intercept covariates as $4$ ($d = p = 5$), and considered three covariate designs: Scenario~1: all $4$ covariates were drawn from independent standard normal distributions; Scenario~2: all $4$ covariates were drawn from independent $\text{Bernoulli}(0.5)$ distributions; Scenario~3: the first two non-intercept covariates were drawn from independent standard normal distributions, and the remaining ones were from independent $\text{Bernoulli}(0.5)$.
The sample size was $2,000$.
The data-generating mechanism and sampling setup are described in Appendix~\ref{apx:numerical_details}.
We developed the P\'{o}lya-Gamma Gibbs sampler with replica exchange, which is detailed in Appendix~\ref{apx:sampling_algorithm} \citep{polson2013bayesian, swendsen1986replica}.

Figure~\ref{fig:posterior_pca} displays the principal component analysis (PCA) plots of the posterior samples after $k$-means++ clustering with $k=2$ \citep{arthur2007kmeans}.
Under continuous and mixed designs (Scenarios~1 and 3), the posterior exhibited clear bimodality, with two well-separated clusters.
Under the binary design (Scenario~2), in which all non-intercept covariates take values in the bounded discrete set $\{0,1\}$, the two clusters were not clearly distinguished in the PCA plots.
This result is consistent with the failure of the binary design to satisfy the condition (C2), leading to a lack of guaranteed identifiability by Corollary~\ref{cor:quotient_mle}.
Posterior means of each cluster and further numerical details are provided in Appendix~\ref{apx:numerical_details} and \ref{apx:trace_hist_posterior}.

\begin{figure}[htbp]
    \centering
    \begin{subfigure}[t]{0.31\textwidth}
        \centering
        \includegraphics[width=\linewidth]{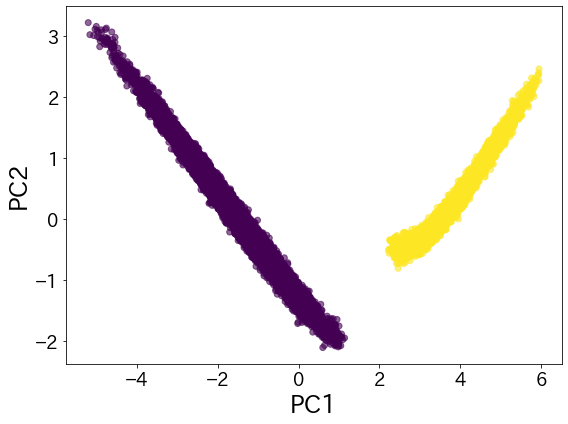}
        \caption{Scenario 1: continuous}
        \label{fig:fig1}
    \end{subfigure}
    \hfill
    \begin{subfigure}[t]{0.31\textwidth}
        \centering
        \includegraphics[width=\linewidth]{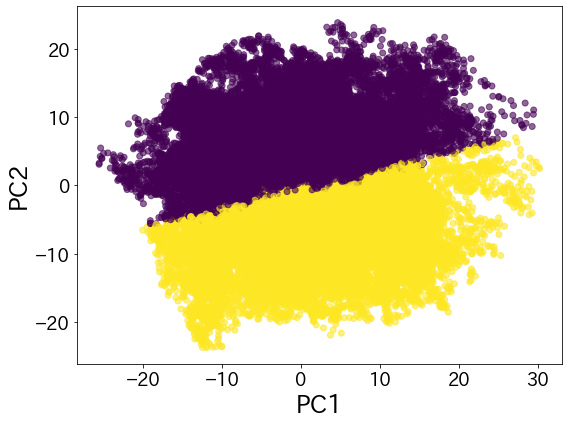}
        \caption{Scenario 2: binary}
        \label{fig:fig2}
    \end{subfigure}
    \hfill
    \begin{subfigure}[t]{0.31\textwidth}
        \centering
        \includegraphics[width=\linewidth]{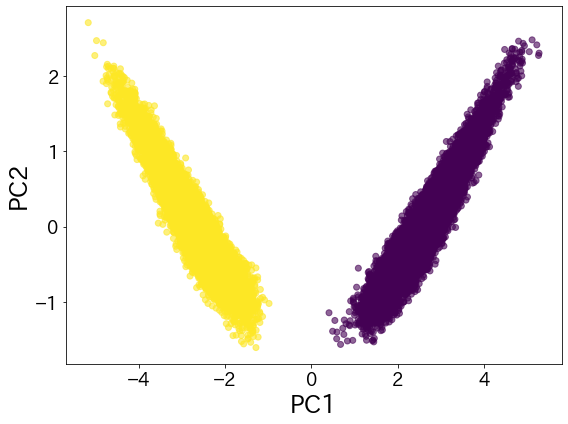}
        \caption{Scenario 3: mixed}
        \label{fig:fig3}
    \end{subfigure}
    \caption{PCA plots of posterior samples after $k$-means clustering ($k=2$). Each panel corresponds to a different covariate scenario: (a) Scenario~1: all 5 covariates were drawn from independent standard normal distributions; (b) Scenario~2: all 5 covariates were drawn from independent $\text{Bernoulli}(0.5)$ distributions; (c) Scenario~3: the first two non-intercept covariates were drawn from independent standard normal distributions, and the remaining ones were from independent $\text{Bernoulli}(0.5)$.}
    \label{fig:posterior_pca}
\end{figure}

\section{Relabeling Rule}
\label{sec:relabeling}

Section~3 shows that, under the shared-design setting, the likelihood identifies only the equivalence class $[\beta,\gamma]$.
In many applications, however, investigators still need a single ordered pair because subsequent interpretation is conducted in terms of the associations with event occurrence or zero-inflation such as zeros due to misclassification.
For that practical purpose, we introduce a simple relabeling rule.

\subsection{Proposed Rule}

We propose a simple relabeling rule.
We proceed as follows.

\begin{enumerate}
    \item[1.] Fit a standard logistic regression of $y$ on $x$ and denote the resulting coefficient vector by $\hat\beta_{\mathrm{LR}}$.
    \item[2.] Fit the zero-inflated logistic regression model with shared design and obtain one ordered maximizer $(\hat\beta,\hat\gamma)$.
    \item[3.] Form the exchange-symmetric solution $(\hat\gamma,\hat\beta)$.
    \item[4.] Choose the pair whose first component vector is closer to $\hat\beta_{\mathrm{LR}}$; that is, choose
    \begin{align*}
        (\hat\beta^{\dagger},\hat\gamma^{\dagger})
        :=
        \argmin_{(\beta,\gamma)\in\{(\hat\beta,\hat\gamma),(\hat\gamma,\hat\beta)\}}
        \|\beta-\hat\beta_{\mathrm{LR}}\|_2
        .
    \end{align*}
\end{enumerate}

We should note that Theorem~\ref{thm:sign-flip} shows that the estimator of the ordinary logistic regression is biased under misspecification.
Therefore, while $\hat\beta_{\mathrm{LR}}$ serves as a convenient reference, it should not be interpreted as a definitive benchmark for the ordinary logistic regression component of the zero-inflated model.
Whenever prior knowledge is available to distinguish the two components, that information should take precedence over this rule.

\subsection{Simulation Study}
\label{sec:simulation}

We next investigate the behavior of the proposed relabeling rule through a simulation study.

\paragraph{Simulation Design.}

We considered four scenarios that differ only in the intercept of the structural-zero component, hence, in the probability of zero inflation.
The coefficient vector for the ordinary logistic regression component was fixed at
\begin{align*}
    {\beta}^{*}
    =
    (0.5, ~1.0, ~0.5, ~0.5, ~0.25)^\top
    ,
\end{align*}
and the structural-zero coefficient was
\begin{align*}
    {\gamma}^{*}
    =
    (\gamma_{0,\mathrm{int}}, ~-1.0, ~-1.0, ~0.5, ~0.5)^\top
    .
\end{align*}
We used four values of $\gamma_{0,\mathrm{int}}$.
Specifically, we considered: (i) Very Low Mislabel ($\gamma_{0,\text{int}} = 4.3$), yielding approximately $3.8\%$ structural zeros; (ii) Low Mislabel ($\gamma_{0,\text{int}} = 3.0$), yielding approximately $10.3\%$ structural zeros; (iii) Moderate Mislabel ($\gamma_{0,\text{int}} = 1.7$), yielding approximately $23.4\%$ structural zeros; and (iv) High Mislabel ($\gamma_{0,\text{int}} = 1.0$), yielding approximately $33.4\%$ structural zeros.
These scenarios were determined to span the range of zero-inflation levels commonly encountered in applications, from nearly negligible to substantial proportions of structural zeros.

For each scenario, we generated $n=1,000$ observations with $d=5$ covariates including an intercept.
The first element of $x_i$ was $1$, and the remaining elements were drawn independently from the standard normal distributions.
Responses were generated from the zero-inflated logistic regression model as follows:
\begin{align*}
    h_i &\sim \operatorname{Bernoulli}\left(F({\gamma}^{* \top} x_i)\right)
    ,
    \\
    \quad
    y_i^\ast &\sim \operatorname{Bernoulli}\left(F({\beta}^{* \top} x_i)\right)
    ,
    \\
    \quad
    y_i
    &= h_i y_i^\ast
    .
\end{align*}

We compared three estimation approaches:
(i) Proposed approach: the zero-inflated model with the proposed relabeling rule in Section \ref{sec:relabeling}, (ii) The standard logistic regression approach: the ordinary logistic regression model ignoring zero inflation, and (iii) Naive zero-inflated model approach: the zero-inflated model without relabeling, retaining the first local maximizer returned by the optimization algorithm.

All methods were performed by the same optimization method and settings: the L-BFGS-B algorithm with analytical gradients, a maximum of $1,000$ iterations, and random initialization from $\mathcal{N}(0,0.01I)$.
An estimate was classified as unreasonable if at least one component exceeds ten times the absolute value of its true value.
Each scenario was replicated $10,000$ times.

\paragraph{Results.}

Tables~\ref{tab:simulation_results}--\ref{tab:convergence} summarize the results.
Figure~\ref{fig:bias_boxplots} and Figure~\ref{fig:gamma_bias_boxplots} visualize the boxplots for distributions of the observed bias.
In these results, several patterns were clearly observed.
First, for the parameter $\beta$, the bias of the estimates from the proposed approach remained more concentrated around zero than that of the estimates from the naive approach, whereas the estimates from standard logistic regression approach exhibited a systematic negative shift that became larger as the zero-inflation proportion increased.
Second, the proposed approach was more stable than the naive approach.
Across the low-to-moderate zero-inflation scenarios, the biases of the relabeled estimates were smaller than those of the standard logistic regression, while their magnitudes of standard deviations remained moderate.
For for the parameter $\gamma$, the proposed approach also dominated the naive approach in most scenarios, especially for the parameters other than the intercept.
A weakness appeared when structural zeros were very rare, probably because the intercept of the structural-zero component was weakly identified and highly variable.
Third, the naive approach behaved as expected in terms of a non-identifiable ordered parameterization.
Specifically, the estimates were obtained from both symmetric solutions.
The concentration of the relabeled estimates around the true parameter values, in contrast to the spread of the naive estimates, suggests that the proposed relabeling rule was effective at selecting an appropriate representative from each equivalence class.

Table~\ref{tab:convergence} shows that the proposed relabeling rule improved practical reliability.
The proposed method produced reasonable estimates in more than $99\%$ of replications in all but the Very Low Mislabel scenario and consistently outperformed the naive zero-inflated model approach in terms of reasonable solutions.

\begin{table}[htbp]
{\scriptsize
\centering
\caption{Simulation results based on $10{,}000$ replications: estimates of the parameters for the ordinary logistic regression component ($\beta$)}
\label{tab:simulation_results}
\begin{tabular}{llrrrrr}
\toprule
\multirow{2}{*}{Scenario} & \multirow{2}{*}{Method} & \multicolumn{5}{c}{Bias (SD)} \\
\cmidrule{3-7}
& & $\beta_0$ & $\beta_1$ & $\beta_2$ & $\beta_3$ & $\beta_4$ \\
\midrule
\multirow{3}{*}{Very Low Mislabel} & Proposed & 0.102 (0.234) & -0.007 (0.231) & -0.003 (0.190) & -0.004 (0.114) & -0.002 (0.115) \\
 & Standard LR & -0.170 (0.071) & -0.217 (0.083) & -0.166 (0.076) & 0.005 (0.076) & 0.032 (0.072) \\
 & Naive ZILR & 1.100 (1.449) & -0.521 (0.861) & -0.395 (0.674) & -0.004 (0.243) & 0.064 (0.253) \\
\midrule
\multirow{3}{*}{Low Mislabel} & Proposed & 0.055 (0.253) & -0.019 (0.274) & -0.013 (0.218) & -0.001 (0.109) & 0.001 (0.109) \\
 & Standard LR & -0.409 (0.068) & -0.455 (0.075) & -0.346 (0.069) & 0.011 (0.073) & 0.064 (0.069) \\
 & Naive ZILR & 1.175 (1.384) & -0.882 (1.059) & -0.667 (0.810) & 0.010 (0.184) & 0.119 (0.231) \\
\midrule
\multirow{3}{*}{Moderate Mislabel} & Proposed & -0.005 (0.289) & -0.267 (0.592) & -0.207 (0.457) & 0.014 (0.121) & 0.043 (0.142) \\
 & Standard LR & -0.815 (0.068) & -0.724 (0.069) & -0.559 (0.068) & 0.032 (0.072) & 0.115 (0.070) \\
 & Naive ZILR & 0.751 (0.953) & -1.015 (1.127) & -0.767 (0.862) & 0.010 (0.167) & 0.135 (0.225) \\
\midrule
\multirow{3}{*}{High Mislabel} & Proposed & -0.195 (0.284) & -0.648 (0.780) & -0.506 (0.597) & 0.035 (0.132) & 0.105 (0.162) \\
 & Standard LR & -1.121 (0.070) & -0.860 (0.068) & -0.670 (0.069) & 0.050 (0.074) & 0.145 (0.072) \\
 & Naive ZILR & 0.416 (0.802) & -1.021 (1.146) & -0.771 (0.867) & 0.007 (0.181) & 0.133 (0.225) \\
\bottomrule
\end{tabular}
}
\end{table}

\begin{table}[htbp]
{\scriptsize
\centering
\caption{Simulation results based on $10{,}000$ replications: estimates of the parameters for the structural-zero component ($\gamma$)}
\label{tab:simulation_results_gamma}
\begin{tabular}{llrrrrr}
\toprule
\multirow{2}{*}{Scenario} & \multirow{2}{*}{Method} & \multicolumn{5}{c}{Bias (SD)} \\
\cmidrule{3-7}
& & $\gamma_0$ & $\gamma_1$ & $\gamma_2$ & $\gamma_3$ & $\gamma_4$ \\
\midrule
\multirow{2}{*}{Very Low Mislabel} & Proposed & 1.144 (3.760) & -0.102 (1.478) & -0.140 (1.297) & 0.092 (0.674) & 0.086 (0.692) \\
 & Naive ZILR & -0.579 (3.822) & 0.677 (1.581) & 0.471 (1.324) & 0.060 (0.553) & -0.042 (0.579) \\
\midrule
\multirow{2}{*}{Low Mislabel} & Proposed & 0.548 (1.899) & -0.118 (0.883) & -0.109 (0.733) & 0.047 (0.339) & 0.042 (0.348) \\
 & Naive ZILR & -0.802 (2.053) & 0.852 (1.262) & 0.632 (0.977) & 0.028 (0.270) & -0.094 (0.309) \\
\midrule
\multirow{2}{*}{Moderate Mislabel} & Proposed & 0.345 (0.957) & 0.223 (1.070) & 0.160 (0.834) & 0.007 (0.215) & -0.018 (0.268) \\
 & Naive ZILR & -0.432 (1.057) & 0.981 (1.150) & 0.728 (0.882) & 0.011 (0.176) & -0.113 (0.228) \\
\midrule
\multirow{2}{*}{High Mislabel} & Proposed & 0.553 (0.789) & 0.629 (1.350) & 0.484 (1.030) & -0.021 (0.219) & -0.092 (0.278) \\
 & Naive ZILR & -0.071 (0.859) & 1.005 (1.163) & 0.752 (0.882) & 0.006 (0.181) & -0.120 (0.228) \\
\bottomrule
\end{tabular}
}
\end{table}

\begin{figure}[htbp]
\centering
\includegraphics[width=\textwidth]{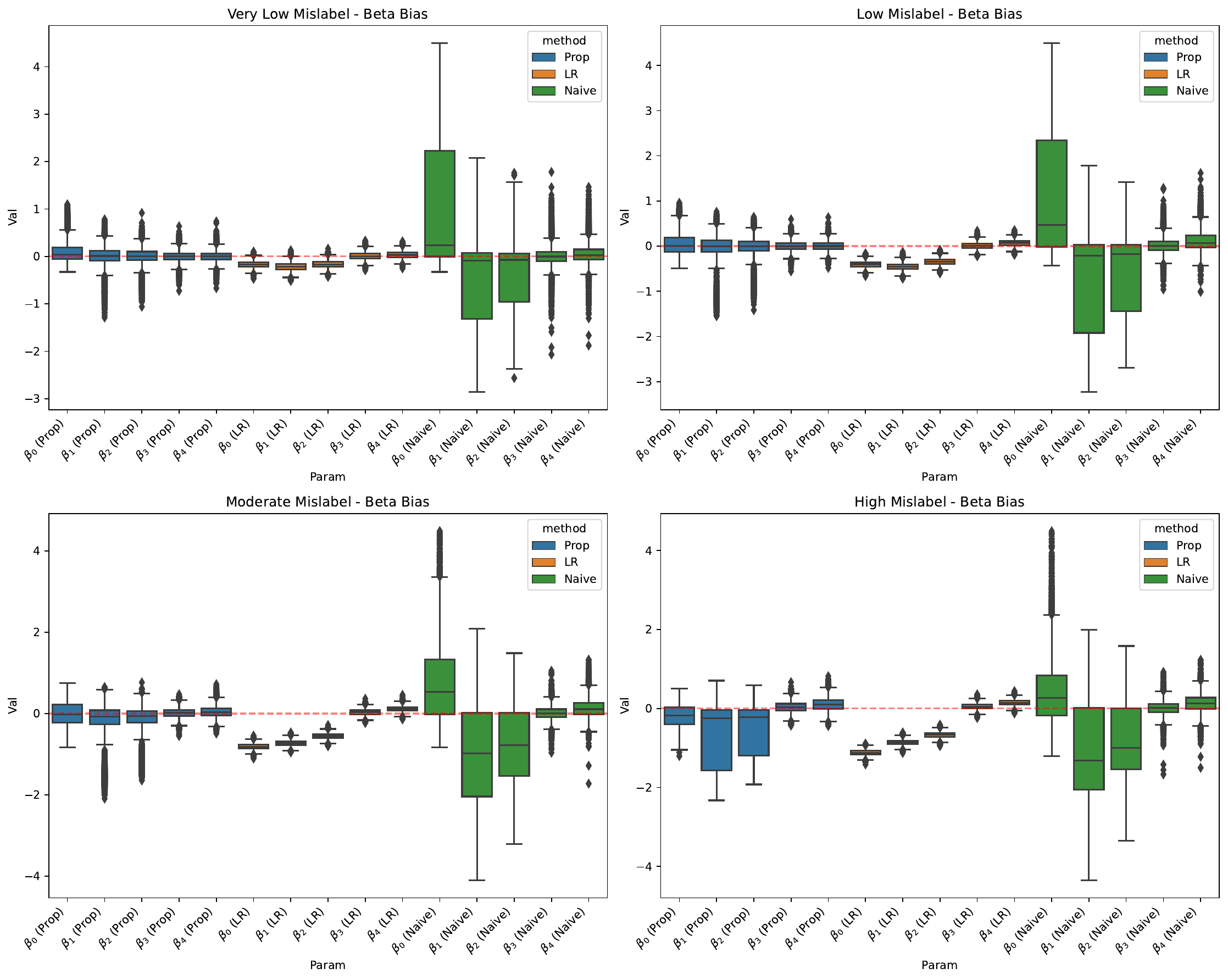}
\caption{Boxplots of parameter bias (estimate minus true value) for $\beta$ parameters across four scenarios based on $10{,}000$ simulation replications. Each panel represents a different scenario, with three estimation methods compared within panels. The true parameter values were set as ${\beta}_0 = (0.5, ~1.0, ~0.5, ~0.5, ~0.25)^\top$.}
\label{fig:bias_boxplots}
\end{figure}

\begin{figure}[htbp]
\centering
\includegraphics[width=\textwidth]{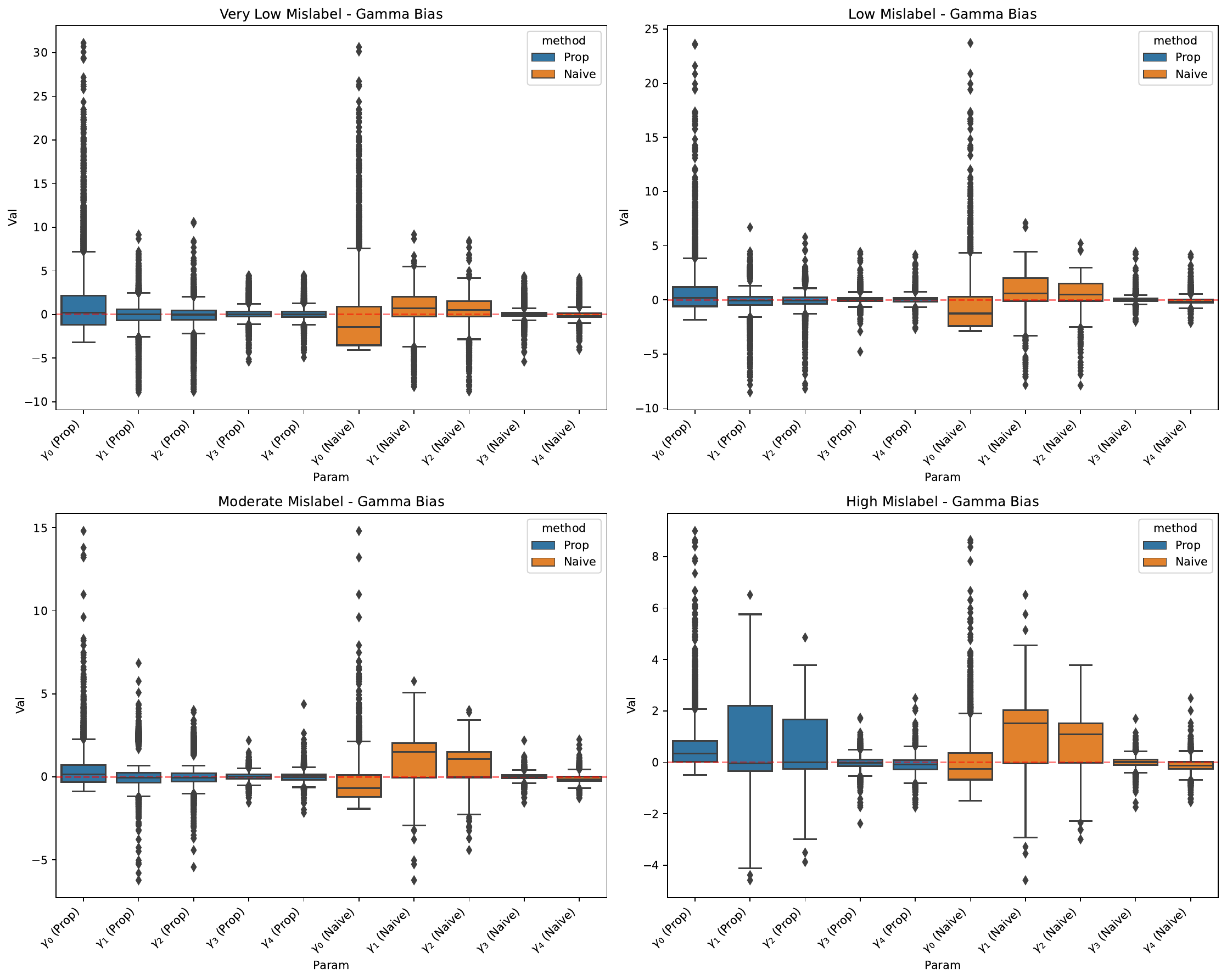}
\caption{Boxplots of parameter bias (estimate minus true value) for $\gamma$ parameters across four scenarios based on $10,000$ simulation replications. Each panel represents a different scenario, with three estimation methods compared within panels. The true parameter values were set as ${\gamma}_0 = (\gamma_{0,\mathrm{int}}, ~-1.0, ~-1.0, ~0.5, ~0.5)^\top$, where $\gamma_{0,\mathrm{int}}=4.3,~3.0,~1.7,~1.0$ for each scenario.}
\label{fig:gamma_bias_boxplots}
\end{figure}

\begin{table}[htbp]
\centering
{\scriptsize
\caption{Convergence diagnostics across simulation scenarios based on $10,000$ replications. "Converged" indicates successful optimization with reasonable parameter estimates, "Unreasonable" denotes cases where at least one parameter estimate exceeded ten times its true value, and "Ratio" shows the percentage of successful convergence.}
\label{tab:convergence}
\begin{tabular}{llrrrr}
\toprule
Scenario & Method & Converged & Unreasonable & Total & Ratio (\%) \\
\midrule
\multirow{3}{*}{Very Low Mislabel} 
& Proposed & 9,156 & 844 & 10,000 & 91.6 \\
& Standard LR & 10,000 & 0 & 10,000 & 100.0 \\
& Naive ZILR & 7,249 & 2,751 & 10,000 & 72.5 \\
\midrule
\multirow{3}{*}{Low Mislabel} 
& Proposed & 9,936 & 64 & 10,000 & 99.4 \\
& Standard LR & 10,000 & 0 & 10,000 & 100.0 \\
& Naive ZILR & 9,232 & 768 & 10,000 & 92.3 \\
\midrule
\multirow{3}{*}{Moderate Mislabel} 
& Proposed & 9,985 & 15 & 10,000 & 99.9 \\
& Standard LR & 10,000 & 0 & 10,000 & 100.0 \\
& Naive ZILR & 9,936 & 64 & 10,000 & 99.4 \\
\midrule
\multirow{3}{*}{High Mislabel} 
& Proposed & 9,985 & 15 & 10,000 & 99.9 \\
& Standard LR & 10,000 & 0 & 10,000 & 100.0 \\
& Naive ZILR & 9,955 & 43 & 10,000 & 99.6 \\
\bottomrule
\end{tabular}
}
\end{table}

\section{Application to Actual Data}
\label{sec:application}

We illustrate the performance of the zero-inflated logistic model with shared design through an application to public data.
We note that the data analysis in this section is intended only as a methodological illustration and not for clinical interpretation.

\subsection{Dataset}

We used the National Health and Nutrition Examination Survey (NHANES), the 2017--2018 public release \citep{nhanes}.
The outcome $y_i$ for individual ID (SEQN) $i$ was self-reported diabetes status, constructed from the diabetes questionnaire as $y_i = 1$ for respondents who reported a prior diabetes diagnosis ($\text{DIQ010} = 1$) and $y_i = 0$ otherwise ($\text{DIQ010} = 2$).
As covariates, we used insurance coverage (HIQ011), usual source of care (HUQ030), age (RIDAGEYR), body mass index: bmi (BMXBMI), and sex (RIAGENDR).
To motivate a zero-inflated model, we compared self-reported status with an HbA1c-based variable.
Specifically, we defined $d_{\mathrm{A1c},i} = 1$ when HbA1c was at least $6.5\%$ ($\text{LBXGH} \geq 6.5$) and $d_{\mathrm{A1c},i} = 0$ otherwise ($\text{LBXGH} < 6.5$).
We used samples with non-zero 2-year sample weights $({\text{WTMEC2YR} >0})$.
Table \ref{tab:nhanes_summary} summarizes the sample size of the data.
Among respondents with $d_{\mathrm{A1c},i}=1$ and non-missing self-report, the proportion with $y_i=0$ was about $21.4\%$.
Although HbA1c is not a gold standard for diagnosis, these descriptive proportions suggest that self-reported diabetes can contain a non-negligible proportion of undiagnosed cases.

\begin{table}[htbp]
\centering
\caption{Sample size summaries of NHANES data used to illustrate zero-inflated logistic regression. The last column reports the proportion of self-reported non-cases ($y_i = 0$) among respondents with HbA1c $\geq 6.5\%$ ($d_{\mathrm{A1c},i} = 1$) and observed self-reported diabetes status.}
\label{tab:nhanes_summary}
\begin{tabular}{cccccc}
\toprule
Period & \shortstack{Interviewed\\participants} & \shortstack{sample weights\\$> 0$} & \shortstack{HbA1c was\\observed} & \shortstack{Self report was\\observed} & \shortstack{Ratio of $y_i = 0$\\given $d_{\mathrm{A1c},i} = 1$} \\
\midrule
2017--2018 & 9,254 & 8,704 & 6,045 & 8,709 & 0.214 \\
\bottomrule
\end{tabular}
\end{table}

\subsection{Model Settings}

We specified a shared-design model.
Specifically, the covariates of the ordinary logistic regression component and the structural-zero component were set equal, with covariates
\begin{align*}
    (1,~\text{insured},~\text{usualcare},~\text{age},~\text{bmi},~\text{female})
    ,
\end{align*}
where age and bmi were standardized.
The model was estimated on the complete-case subset with $711$ samples.
Because the aim of this section is methodological illustration, we did not incorporate survey weights.

\subsection{Results}

We obtained one ordered solution using the BFGS algorithm from a random initial value and then obtained the second by exchanging the two coefficient vectors.
Table \ref{tab:nhanes_modes} shows the two solutions: solution A and B, denoted as $(\hat\beta_{\mathrm{A}}, \hat\gamma_{\mathrm{A}})$ and $(\hat\beta_{\mathrm{B}}, \hat\gamma_{\mathrm{B}})$, respectively.
The resulting negative log-likelihood values were identical up to numerical precision.
Specifically, both of the evaluated values were $1918.023$.
Moreover, the exchange symmetry was numerically exact: $\max_{j \in \{1,\ldots,6\}} |\hat\beta_{A,j}-\hat\gamma_{B,j}| = 2.98\times 10^{-7}$ and $\max_{j \in \{1,\ldots,6\}} |\hat\gamma_{A,j}-\hat\beta_{B,j}| = 7.40\times 10^{-7}$, where $\hat\beta_{A,j}$ denotes the $j$th element of $\hat\beta_{A}$.

We applied the relabeling rule from Section~\ref{sec:relabeling}.
Then, solution B was selected because $\|\hat\beta_{\mathrm{B}}-\hat\beta_{\mathrm{LR}}\|_2^2 = 1.102$, and $\|\hat\beta_{\mathrm{A}}-\hat\beta_{\mathrm{LR}}\|_2^2 = 29.771$, where $\hat\beta_{\mathrm{LR}}$ denotes the estimate from the ordinary logistic regression, shown in Table \ref{tab:nhanes_modes}.

\begin{table}[htbp]
\centering
\caption{Two solutions $(\hat\beta_{\mathrm{A}}, \hat\gamma_{\mathrm{A}})$ and $(\hat\beta_{\mathrm{B}}, \hat\gamma_{\mathrm{B}})$, and the estimate from the ordinary logistic regression $\hat\beta_{\mathrm{LR}}$ for the NHANES illustration.}
\label{tab:nhanes_modes}
\begin{tabular}{lrrrrr}
\toprule
Term & $\hat\beta_{\mathrm{A}}$ & $\hat\gamma_{\mathrm{A}}$ & $\hat\beta_{\mathrm{B}}$ & $\hat\gamma_{\mathrm{B}}$ & $\hat\beta_{\mathrm{LR}}$\\
\midrule
Intercept &  1.250 & -3.192 & -3.192 &  1.250 & -3.444\\
insured   & -0.413 &  0.138 &  0.138 & -0.413 & -0.111\\
usualcare &  1.085 &  0.179 &  0.179 &  1.085 & 0.627\\
age   & -1.251 &  2.300 &  2.300 & -1.251 & 1.474\\
bmi   &  0.647 &  0.371 &  0.371 &  0.647 & 0.590\\
female    & -0.500 & -0.223 & -0.223 & -0.500 & -0.431\\
\bottomrule
\end{tabular}
\end{table}

\section{Discussion}
\label{sec:discussion}

This study investigates the zero-inflated logistic regression model with shared design, in terms of a sign-flip phenomenon under misspecification, the existence of maximum likelihood estimates, identifiability of the regression parameters, computational methods for implementation, and a practical relabeling rule. 
The primary theoretical message is that non-identifiability in the shared-design setting is not unstructured.
Under mild regularity conditions, the non-identifiability is reduced to the exchange symmetry of the two coefficient vectors.
By considering the quotient space with respect to this symmetry, the expected log-likelihood has a unique maximizer.
This result is useful for understanding the inherent inferential limits of the model.

A second contribution is the analysis of the existence of maximum likelihood estimates.
The concepts of double separation and $\varepsilon$--double--non-separation introduced in this study extend the classical separation conditions for the ordinary logistic regression model.
While these conditions do not provide a complete characterization, they offer tractable sufficient conditions for both the existence and non-existence of estimates.
In particular, the results on non-existence explain why optimization algorithms may fail to converge even before considering the exchange symmetry of the regression parameters.
Furthermore, as shown in Theorem~\ref{thm:sign-flip}, model misspecification can lead to a sign flip in the regression coefficients relative to their true values.
This provides a formal warning against analyses that ignore zero-inflation structures.

Numerical results based on posterior sampling support the theoretical findings regarding identifiability.
Bimodality was clearly observed in posterior distributions under the continuous and mixed designs.
However, the binary design did not exhibit clear mode separation and increased numerical instability, probably because the design fails to satisfy the condition (C2) and thereby lacks the guaranteed identifibility.
These results provide a practical guideline: analysts should exercise caution when interpreting estimates if the all covariate values are restricted to a small number of support points in the covariate space.
In terms of sampling algorithm, because standard single-chain Gibbs samplers may become trapped in one of the modes, we employed replica exchange method for efficient exploration of the parameter space.

The relabeling rule proposed in Section~\ref{sec:relabeling} serves as a heuristic for interpretation rather than a new source of identification.
Its role is to provide a reproducible ordering rule when an ordered pair is required for subsequent interpretation or comparison with the results from the ordinary logistic regression.
If external information is available to distinguish the ordinary logistic regression component from the structural-zero component, such information should take precedence over the proposed rule.

This study has several limitations. 
First, our theoretical results provide sufficient conditions for the existence and non-existence of the maximum likelihood estimate, rather than a complete characterization. 
Second, regarding the relabeling rule, it remains to be investigated whether alternative rules can improve performance when the referenced ordinary logistic regression itself suffers from a severe bias.
Third, as the asymptotic theory for the model on the quotient space remains to be established, we do not perform formal statistical inference. 
These limitations suggest several directions for future research. 
Promising directions include a more refined characterization of the existence of the estimates, the formalization of asymptotic theory for parameters defined on the quotient space, and the development of relabeling rules that provide reliable choice even when the referenced logistic regression suffers from a severe bias.

A further direction is to examine whether the theoretical results extend to other link functions for binary regression, such as the probit or the complementary log-log function.
Theorem~\ref{thm:two_mode} exploits the specific logistic form $F(\mu)=\exp(\mu)/\{1+\exp(\mu)\}$, which reduces the identifiability condition to an equality of sums of exponential terms with distinct exponents. 
While such a representation is unavailable for other general link functions, the principle that the equality $F(\beta_1^\top x)F(\gamma_1^\top x)=F(\beta_2^{\top}x)F(\gamma_2^{\top}x)$ over an open set restricts the parameter pairs to an exchange-symmetric set may hold for a broader class of functions.
Specifically, for real analytic link functions, an exchange-symmetry result may be established via the identity theorem, provided that the analytic form of the product function ensures that local equality implies global equivalence.
Furthermore, because the results on the existence of the maximum likelihood estimate and the sign flip phenomenon under misspecification depend primarily on the monotonicity and boundedness of $F(\cdot)$ and the structure of the log-likelihood function, these properties are expected to hold across other link functions, probably with appropriate modifications to reflect specific characteristics of link functions, such as the asymmetric behavior of the complementary log-log function.

\section{Concluding Remarks}
\label{sec:conclusion}

When the same covariates are used for both components of the zero-inflated logistic regression model, identifiability results as established in existing literatures does not hold.
We establish that this non-identifiability has a specific structure.
Namely, under mild regularity conditions, it reduces to exchange symmetry, and the expected log-likelihood has a unique maximizer on the resulting quotient space.
In addition, we introduce sufficient conditions for the existence and non-existence of the maximum likelihood estimate, demonstrate posterior bimodality through numerical experiments, and propose a simple relabeling rule for applications.
We also establish a sign flip phenomenon under misspecification.
These theoretical and numerical results provide us with a practical guideline in applying the zero-inflated logistic regression model.


\section*{Funding}

This research was supported by AMED under Grant Number JP223fa627001 (UTOPIA AI Research Discovery Program) and JSPS KAKENHI Grant Number 26K02664.

\section*{Data Availability}

The NHANES data used in the application is available from the website of the Centers for Disease Control and Prevention: \url{https://wwwn.cdc.gov/nchs/nhanes/}.

\section*{Code Availability}

The Python scripts for numerical studies and the R script for actual data application in this manuscript are available from the GitHub repository: \url{https://github.com/t-yui/zero-inflated-logistic-shared-design}.

\section*{Declaration of the Use of Generative AI and AI-assisted Technologies}

The authors used ChatGPT (OpenAI), Claude (Anthropic) and Gemini (Google) to assist with developing scripts for simulation and application, and editing the English language during the preparation of this manuscript.
The authors checked and edited the content and take full responsibility for this manuscript.

\appendix

\section{Proofs}

\subsection{Proof of Theorem~\ref{thm:sign-flip}}
\label{apx:proof_misspec}

We state the following regularity conditions.

\begin{assumption}\label{ass:basic}
We assume that
\begin{enumerate}
    \item[(A1)] $\mathbb{E}[\|x\|^2] < \infty$.
    \item[(A2)] $\mathbb{P}(\tilde{x}_j>0)>0$ and $\mathbb{P}(\tilde{x}_j<0)>0$.
    \item[(A3)] $\mathbb{E}[\tilde{x}_j\mid \tilde{x}_{-j},\tilde{z}_{-j}]=0$ almost surely.
    \item[(A4)] For the fixed value of $c<0$ and every $t\in\mathbb{R}$, the map $(\theta_0,\theta_{-j}) \mapsto \mathcal{L}_c(\theta_0,t,\theta_{-j})$ has a unique maximizer, and the profile objective function $g_c$ has at least one maximizer on $\mathbb{R}$.
\end{enumerate}
\end{assumption}

For the fixed value of $c<0$ and each $t\in\mathbb{R}$, let $(\theta_0^{*}(t),\theta_{-j}^{*}(t))$ denote the maximizer in the definition of $g_c(t)$, and let
\begin{align*}
    \eta(t,x)
    &:=
    \theta_0^{*}(t) + t \tilde{x}_j + \theta_{-j}^{*}(t)^{\top}\tilde{x}_{-j},
    \\
    \mu(t,x)
    &:=
    F(\eta(t,x))
    .
\end{align*}

We first establish basic properties of the profile objective function $g_c$.

\begin{lemma}
\label{lem:g_derivative}
Suppose that Assumption~\ref{ass:basic} holds. Then, we have
\begin{enumerate}
\item[(i)] 
The derivative of $g_c$ is $
    g_c'(t)
    = \mathbb{E}_{(x,z)}\left[(\pi_c(x,z)-\mu(t,x))\,\tilde{x}_j\right]
$.
\item[(ii)]
The function $g_c$ is concave with respect to $t$.
\end{enumerate}
\end{lemma}

\begin{proof}[Proof of Lemma~\ref{lem:g_derivative}.]
We begin with (i).
For fixed $(x,z)$, let
\begin{align*}
    \ell(\eta)
    &:=
    \pi_c(x,z)\log F(\eta)
    + \{1-\pi_c(x,z)\}\log\{1-F(\eta)\}
    .
\end{align*}
Using $F'(\eta)=F(\eta)\{1-F(\eta)\}$, we have
\begin{align*}
    \frac{\partial}{\partial \eta}\ell(\eta)
    &=
    \pi_c(x,z)-F(\eta)
    .
\end{align*}
Since $\partial\eta/\partial t = x_j$, we obtain
\begin{align*}
    \frac{\partial}{\partial t}\mathcal{L}_c(\theta_0,t,\theta_{-j})
    &= \mathbb{E}_{(x,z)}\left[
       \frac{\partial \ell(\eta)}{\partial \eta}\,
       \frac{\partial \eta}{\partial t}
     \right]
    = \mathbb{E}_{(x,z)}\left[(\pi_c(x,z)-F(\eta))\,x_j\right]
    .
\end{align*}
Moreover, since $0\leq \pi_c(x,z)\leq 1$ and $0\leq F(\eta)\leq 1$, we have
\begin{align*}
    \left|
    \frac{\partial}{\partial t}\ell(\eta)
    \right|
    \leq
    |\tilde{x}_j|
    .
\end{align*}
By Assumption~\ref{ass:basic}(A1), we have $\mathbb{E}[|\tilde{x}_j|] \leq \{\mathbb{E}[\tilde{x}_j^2]\}^{1/2} < \infty$.
Therefore, by the uniqueness of the maximizer from Assumption~\ref{ass:basic}(A4), using Danskin's theorem, we obtain
\begin{align*}
    g_c'(t)
    &=
    \frac{\partial}{\partial t}\mathcal{L}_c(\theta_0,t,\theta_{-j})
    \Big|_{(\theta_0,\theta_{-j})
          =(\theta^{*}_0(t),\theta^{*}_{-j}(t))}
    \\
    &=
    \mathbb{E}_{(x,z)}\left[
    (\pi_c(x,z)-\mu(t,x))\tilde{x}_j
    \right]
    .
\end{align*}

We then prove (ii).
Fix $t_1,~t_2\in\mathbb{R}$ and $\lambda\in[0,1]$.
Let $(\theta^{*}_{0,r},\theta^{*}_{-j,r}) \in \mathbb{R} \times \mathbb{R}^{d-2}$ be a maximizer of $\mathcal{L}_c(\cdot,t_r,\cdot)$ for $r=1,~2$.
Using the concavity of $\mathcal{L}_{c}$, we have
\begin{align*}
    &g_c(\lambda t_1+(1-\lambda)t_2)
    \\
    &\quad=
    \sup_{(\theta_{0},\theta_{-j}) \in \mathbb{R} \times \mathbb{R}^{d-2}} \mathcal{L}_c(\theta_{0}, \lambda t_1+(1-\lambda)t_2, \theta_{-j})
    \\
    &\quad\geq
    \mathcal{L}_{c}(\lambda \theta^{*}_{0,1}+(1-\lambda) \theta^{*}_{0,2},~\lambda t_1+(1-\lambda)t_2,~\lambda \theta^{*}_{-j,1}+(1-\lambda) \theta^{*}_{-j,2})
    \\
    &\quad\geq
    \lambda \mathcal{L}_{c}(\theta^{*}_{0,1},~t_1,~\theta^{*}_{-j,1})
    +
    (1 - \lambda) \mathcal{L}_{c}(\theta^{*}_{0,2},~t_2,~\theta^{*}_{-j,2})
    \\
    &\quad=
    \lambda g_c(t_1)+(1-\lambda)g_c(t_2)
    .
\end{align*}
Thus $g_c (t)$ is concave.
\end{proof}

The next lemma simplifies the derivative at $t=0$.

\begin{lemma}
\label{lem:mu-zero}
Suppose that Assumption~\ref{ass:basic} holds.
Then, we have
\begin{align*}
    \mathbb{E}_{(x,z)}\left[(\pi_c(x,z)-\mu(0,x)) \tilde{x}_j\right]
    =
    \mathbb{E}_{(x,z)}\left[\pi_c(x,z) \tilde{x}_j\right]
    .
\end{align*}
\end{lemma}

\begin{proof}[Proof of Lemma~\ref{lem:mu-zero}.]
At $t=0$, we have
$
    \mu(0,x)
    = F(\theta^{*}_0(0)+\theta^{*}_{-j}(0)^{\top}\tilde{x}_{-j})
    =: \tilde{\mu}(\tilde{x}_{-j})
$, which depends on $x$ only through $\tilde{x}_{-j}$.
Thus, using (A3),
\begin{align*}
    \mathbb{E}_{(x,z)}[\mu(0,x)\tilde{x}_j]
    &= \mathbb{E}_{(\tilde{x}_{-j},\tilde{z}_{-j})}\left[
       \tilde{\mu}(\tilde{x}_{-j})\mathbb{E}[\tilde{x}_j\mid \tilde{x}_{-j},\tilde{z}_{-j}]
     \right]
    = 0
    .
\end{align*}
Therefore, we have
\begin{align*}
    \mathbb{E}_{(x,z)}\left[(\pi_c(x,z)-\mu(0,x)) \tilde{x}_j\right]
    =
    \mathbb{E}_{(x,z)}\left[\pi_c(x,z) \tilde{x}_j\right]
    .
\end{align*}
\end{proof}

Finally, we investigate the behavior of $f(c):=\mathbb{E}_{(x,z)}[\pi_c(x,z) \tilde{x}_j]$.

\begin{lemma}
\label{lem:monotone-limit}
Suppose that Assumption~\ref{ass:basic} holds.
Then, we have
\begin{enumerate}
\item[(i)]
For every $c<0$, 
$
    f'(c)
    = \mathbb{E}_{(x,z)}[
        F(\beta^{\top}x)~F'(\gamma^{\top}z) \tilde{x}_j^2
    ] > 0
    .
$
\item[(ii)]
$
    \lim_{c\to-\infty} f(c)
    =
    \mathbb{E}_{(x,z)}[
        F(\beta^{\top}x)~\tilde{x}_j~\mathbf{1}\{\tilde{x}_j<0\}
    ] < 0
    .
$
\end{enumerate}
\end{lemma}

\begin{proof}[Proof of Lemma~\ref{lem:monotone-limit}.]
We first prove (i).
Since $\gamma^{\top}z = \gamma_0 + c \tilde{x}_j + \gamma_{-j}^{\top}\tilde{z}_{-j}$, we have $\partial(\gamma^{\top}z)/\partial c = \tilde{x}_j$, and hence
\begin{align*}
    \frac{\partial}{\partial c}\pi_c(x,z)
    =
    F(\beta^{\top}x) F'(\gamma^{\top}z)\tilde{x}_j
    .
\end{align*}
We then obtain
\begin{align*}
    f'(c)
    = \mathbb{E}_{(x,z)}\left[
        \frac{\partial}{\partial c}\{\pi_c(x,z)\tilde{x}_j\}
    \right]
    = \mathbb{E}_{(x,z)}\left[
        F(\beta^{\top}x)F'(\gamma^{\top}z)\tilde{x}_j^2
    \right]
    .
\end{align*}
Since $F'(\cdot)>0$ and $\tilde{x}_j^2 \geq 0$ with
$\mathbb{P}(\tilde{x}_j \neq 0)>0$ by (A2), the expectation is strictly positive and thus we have $f'(c)>0$.

We next prove (ii).
As $c \to -\infty$, we have
\begin{align*}
    \gamma^{\top}z \to
    \begin{cases}
        +\infty, & \text{if}~ \tilde{x}_j<0
        ,
        \\
        -\infty, & \text{if}~\tilde{x}_j>0
        ,
        \\
        \text{finite}, & \text{if}~\tilde{x}_j=0
        ,
    \end{cases}
\end{align*}
and hence $F(\gamma^{\top}z)\to \mathbf{1}\{\tilde{x}_j<0\}$.
We then have
\begin{align*}
    \lim_{c\to-\infty} f(c)
    = \mathbb{E}_{(x,z)}\left[
        F(\beta^{\top}x) \tilde{x}_j\mathbf{1}\{\tilde{x}_j<0\}
    \right]
    .
\end{align*}
For the event $\{\tilde{x}_j<0\}$, the integrand is strictly negative because $F(\beta^{\top}x)>0$ and $\tilde{x}_j<0$.
Therefore, by (A2), we obtain $\lim_{c\to-\infty} f(c) < 0$.
\end{proof}

We now prove Theorem~\ref{thm:sign-flip}.

\begin{proof}[Proof of Theorem~\ref{thm:sign-flip}]
By Lemma~\ref{lem:g_derivative}(i) and Lemma~\ref{lem:mu-zero}, the derivative
of $g_c$ at $t=0$ is
\begin{align*}
    g_c'(0)
    =
    \mathbb{E}_{(x,z)}\left[(\pi_c(x,z)-\mu(0,x)) \tilde{x}_j\right]
    =
    \mathbb{E}_{(x,z)}\left[\pi_c(x,z) \tilde{x}_j\right]
    =
    f(c)
    .
\end{align*}
By Lemma~\ref{lem:monotone-limit}(i), $f$ is strictly increasing in $c$, and
by Lemma~\ref{lem:monotone-limit}(ii), $\lim_{c\to-\infty} f(c) < 0$.
Hence there exists a constant $C_0<0$ such that
\begin{align*}
    c \leq C_0
    \quad
    \Longrightarrow
    \quad
    f(c) = g_c'(0)<0
    .
\end{align*}

Fix $c\leq C_0$.
By Lemma~\ref{lem:g_derivative}(ii), $g_c$ is concave.
Therefore, for every $t>0$, we have $g_c(t) \leq g_c(0)+t g_c'(0) < g_c(0)$.
This implies that no maximizer of $g_c$ can lie in $[0,\infty)$.
By Assumption~\ref{ass:basic}(A4), we have $\arg\max_{t \in \mathbb{R}} g_c(t) \subset (-\infty,0)$.
\end{proof}

\subsection{Proof of Proposition \ref{prop:MLEnotexist}}
\label{apx:proof_mle_notexist}

We prove Proposition \ref{prop:MLEnotexist} as follows.

\begin{proof}[Proof of Proposition~\ref{prop:MLEnotexist}.]
Assume that the data satisfy double separation, so there exist non-zero vectors
$v\in\mathbb{R}^{d}$ and $w\in\mathbb{R}^{p}$ such that
\begin{align*}
    y_i=1
    &\Longrightarrow
    v^{\top}x_i\geq 0
    ,
    \quad
    w^{\top}z_i\geq 0
    ,
    \\
    y_i=0
    &\Longrightarrow
    v^{\top}x_i\leq 0
    ,
    \quad
    w^{\top}z_i\leq 0
    ,
\end{align*}
and at least one observation satisfies a strict inequality in the sense of the definition of double separation.

Fix $(\beta,\gamma)\in\mathbb{R}^{d}\times\mathbb{R}^{p}$, and define $
    \ell_{\beta,\gamma}(t)
    :=
    L(\beta+t v,\gamma+t w)
$ for $t\geq 0$.
For each $i$, let
\begin{align*}
    a_i(t)
    &=
    F\left(\gamma^{\top}z_i+t w^{\top}z_i\right)
    ,
    \\
    b_i(t)
    &=
    F\left(\beta^{\top}x_i+t v^{\top}x_i\right)
    ,
    \\
    g_i(t)
    &=
    a_i(t)b_i(t)
    .
\end{align*}
Because $F'(u)=F(u)\{1-F(u)\}$, we have
\begin{align*}
    g_i'(t)
    =
    g_i(t)
    \left(
        (1-a_i(t)) w^{\top}z_i
        +(1-b_i(t)) v^{\top}x_i
    \right)
    .
\end{align*}
Therefore, we have
\begin{align*}
    \ell_{\beta,\gamma}'(t)
    &=
    \sum_{i=1}^{n}
    \left\{
        \frac{y_i}{g_i(t)}
        -
        \frac{1-y_i}{1-g_i(t)}
    \right\}
    g_i'(t)
    \\
    &=
    \sum_{i=1}^{n}
    \frac{y_i-g_i(t)}{1-g_i(t)}
    \left(
        (1-a_i(t)) w^{\top}z_i
        +(1-b_i(t)) v^{\top}x_i
    \right)
    .
\end{align*}

If $y_i=1$, then
\begin{align*}
    \frac{y_i-g_i(t)}{1-g_i(t)}
    =
    \frac{1-g_i(t)}{1-g_i(t)}
    =
    1
    ,
\end{align*}
and both $w^{\top}z_i \geq 0$ and $v^{\top}x_i \geq 0$.
Hence, we obtain
\begin{align*}
    \frac{y_i-g_i(t)}{1-g_i(t)}
    \left(
        (1-a_i(t)) w^{\top}z_i
        +(1-b_i(t)) v^{\top}x_i
    \right) \geq 0
    .
\end{align*}
If $y_i=0$, then
\begin{align*}
    \frac{y_i-g_i(t)}{1-g_i(t)}
    =
    -\frac{g_i(t)}{1-g_i(t)}
    <
    0
    ,
\end{align*}
while both $w^{\top}z_i \leq 0$ and $v^{\top}x_i \leq 0$.
Hence, we again obtain
\begin{align*}
    \frac{y_i-g_i(t)}{1-g_i(t)}
    \left(
        (1-a_i(t)) w^{\top}z_i
        +(1-b_i(t)) v^{\top}x_i
    \right)
    \geq
    0
    .
\end{align*}
Moreover, because at least one observation satisfies a strict inequality, we have
\begin{align*}
    \frac{y_i-g_i(t)}{1-g_i(t)}
    \left(
        (1-a_i(t)) w^{\top}z_i
        +(1-b_i(t)) v^{\top}x_i
    \right)
    >
    0
    ,
\end{align*}
for $t\geq 0$, for some $i$.
Indeed, if $y_i=1$ and either $w^{\top}z_i>0$ or $v^{\top}x_i>0$, then
\begin{align*}
    (1-a_i(t)) w^{\top}z_i
    +(1-b_i(t)) v^{\top}x_i
    >
    0
    ,
\end{align*}
because $0<a_i(t),b_i(t)<1$.
Similarly, if $y_i=0$ and either $w^{\top}z_i<0$ or $v^{\top}x_i<0$, then
\begin{align*}
  (1-a_i(t)) w^{\top}z_i
  +(1-b_i(t)) v^{\top}x_i
  <
  0
  .
\end{align*}
Therefore,
\begin{align*}
    \ell_{\beta,\gamma}'(t)>0
    \quad
    \text{for all}~
    t\geq 0
    .
\end{align*}

Since $(\beta,\gamma)$ was arbitrary, every finite parameter point can be improved by moving a small positive amount along the direction $(v,w)$.
Therefore, no finite point can be a maximizer of $L(\beta,\gamma)$.
Therefore, the log-likelihood has no maximizer in $\mathbb{R}^{d}\times\mathbb{R}^{p}$.
\end{proof}

\subsection{Proof of Proposition \ref{prop:MLEexist}}
\label{apx:proof_mle_exist}

We prove Proposition \ref{prop:MLEexist} as follows.

\begin{proof}[Proof of Proposition~\ref{prop:MLEexist}.]
Let $(\beta,\gamma)=t(v,w)$, $t:=\sqrt{\|\beta\|^2+\|\gamma\|^2}\ge 0$, and $\|v\|^2+\|w\|^2=1$.
Because
\begin{align*}
    F(u) &\leq e^u
    \quad \text{and} \quad
    1-F(u) \leq e^{-u}
    ,
\end{align*}
we have
\begin{align}
    \label{eq:tail_bounds_exist}
    F(u_1)F(u_2) \leq e^{u_1+u_2},
    \quad \text{and} \quad
    1-F(u_1)F(u_2) \leq e^{-u_1}+e^{-u_2}
    ,
\end{align}
for all $u_1,u_2\in\mathbb{R}$.

By $\varepsilon$--double--non-separation, for each unit direction $(v,w)$ at least one of the following holds:
\begin{enumerate}
    \item[(a)] there exists $i\in \{i:y_i=1\}$ such that $v^\top x_i+w^\top z_i\leq -\varepsilon$;
    \item[(b)] there exists $i\in \{i:y_i=0\}$ such that $\min\{v^\top x_i,w^\top z_i\}\geq \varepsilon$.
\end{enumerate}

If (a) holds, then for $i\in \{i:y_i=1\}$,
\begin{align*}
    \log F(t v^\top x_i)+\log F(t w^\top z_i)
    &=
    \log\{F(t v^\top x_i)F(t w^\top z_i)\}
    \\
    &\leq
    t(v^\top x_i+w^\top z_i)
    \leq
    -\varepsilon t
    .
\end{align*}
Therefore, we have
\begin{align*}
    L(tv,tw)
    \leq
    -\varepsilon t
    .
\end{align*}
Next, if (b) holds, then for $i\in \{i:y_i=0\}$,
\begin{align*}
    \log(1-F(t v^\top x_i)F(t w^\top z_i))
    &\leq
    \log(e^{-t v^\top x_i}+e^{-t w^\top z_i})
    \\
    &\leq
    \log 2-\varepsilon t
    .
\end{align*}
Therefore, we have
\begin{align*}
    L(tv,tw)
    \leq
    \log 2 
    -\varepsilon t
    .
\end{align*}
Consequently, we obtain
\begin{align}
    \label{eq:coercive_L}
    \sup_{(v,w):\|v\|^2+\|w\|^2=1}L(tv,tw)
    \leq \log 2-\varepsilon t \to -\infty
    \quad \text{as} \quad
    t \to \infty
    .
\end{align}

Since $\log 2 - \varepsilon t \to -\infty$ as $t \to \infty$, there exists a sufficiently large radius $t_0 > 0$ such that $\log 2 - \varepsilon t_0 < L(0, 0)$.
We then define the closed ball 
\begin{align*}
    \mathcal{B}_{t_0}
    =
    \{ (\beta^{\top},\gamma^{\top})^{\top} \in \mathbb{R}^{d+p} : \|(\beta^{\top},\gamma^{\top})^{\top}\| \leq t_0 \}
    .
\end{align*}
As $\mathcal{B}_{t_0}$ is a compact set and $L$ is continuous, $L$ attains its maximum at some point $(\hat{\beta},\hat{\gamma})^{\top} \in \mathcal{B}_{t_0}$. 
Furthermore, for any $(\beta^{\top},\gamma^{\top})^{\top}$ outside this ball, by \eqref{eq:coercive_L}, we have
\begin{align*}
    L(\beta,\gamma) \leq \log 2 - \varepsilon \|(\beta^{\top},\gamma^{\top})^{\top}\|
    <
    \log 2 - \varepsilon t_0 < L(0, 0)
    .
\end{align*}
Since $L(0, 0) \leq L(\hat{\beta},\hat{\gamma})$, it follows that $L(\beta,\gamma) < L(\hat{\beta},\hat{\gamma})$ for all $(\beta^{\top},\gamma^{\top})^{\top} \notin \mathcal{B}_{t_0}$. 
Therefore, $(\hat{\beta},\hat{\gamma})^{\top}$ is the global maximizer of $L({\beta},{\gamma})$.
\end{proof}

\subsection{Proofs of Proposition~\ref{prop:two_mode_c1}, Theorem~\ref{thm:two_mode}, and Corollary~\ref{cor:quotient_mle}} 
\label{apx:proof_two_mode}

We first remark a standard linear-independence property of exponential functions.

\begin{lemma}
\label{lem:exp_linear_indep}
Let $\lambda_1,\ldots,\lambda_m \in \mathbb{R}^{d-1}$ be distinct vectors, and let $\mathcal{U}\subset\mathbb{R}^{d-1}$ contain a nonempty open set.
If
\begin{align*}
    \sum_{k=1}^{m} a_k \exp\left(\lambda_k^\top w\right)
    =
    0
    ,
\end{align*}
for all $w \in \mathcal{U}$, then $a_1=\cdots=a_m=0$.
\end{lemma}

\begin{proof}[Proof of Lemma~\ref{lem:exp_linear_indep}.]
Choose $w_0$ in the interior of $\mathcal{U}$.
Since the hyperplanes
\begin{align*}
    \left\{
        v\in\mathbb{R}^{d-1}:
        (\lambda_k-\lambda_\ell)^\top v=0
        ,~
        k\neq \ell
        ,~
        k,\ell \in \{1,\ldots,m\}
    \right\}
    ,
\end{align*}
do not cover $\mathbb{R}^{d-1}$, there exists $v \in \mathbb{R}^{d-1}$ such that $\lambda_1^\top v,~ \ldots,~ \lambda_m^\top v$ are pairwise distinct.
For all sufficiently small $t$, we have $w_0+t v \in \mathcal{U}$, and hence
\begin{align*}
    0
    =
    \sum_{k=1}^{m} a_k \exp\left(\lambda_k^\top(w_0+t v)\right)
    =
    \sum_{k=1}^{m} a_k \exp\left(\lambda_k^\top w_0\right)\exp\left((\lambda_k^\top v)t\right)
    .
\end{align*}
Since one-dimensional exponential functions with distinct exponents are linearly independent on any open interval, we obtain
\begin{align*}
    a_k \exp\left(\lambda_k^\top w_0\right)
    =
    0
    ,
\end{align*}
for $k=1,\ldots,m$.
Therefore, we have $a_k=0$ for $k=1,\ldots,m$.
\end{proof}

We now prove Proposition~\ref{prop:two_mode_c1}.

\begin{proof}[Proof of Proposition~\ref{prop:two_mode_c1}]
Because \eqref{eq:equiv_prob} holds for all $y\in\{0,1\}$, it is equivalent to
\begin{align*}
    F\left(x^\top\beta_1\right)F\left(x^\top\gamma_1\right)
    =
    F\left(x^\top\beta_2\right)F\left(x^\top\gamma_2\right)
    ,
\end{align*}
for all $x=(1,\tilde{x}_{-0}^\top)^\top$ with $\tilde{x}_{-0}\in\mathcal{U}$.
Using $F(\mu)=1/(1+\exp(-\mu))$, we obtain
\begin{align*}
    &
    \left\{1+\exp(-\beta_{1,0})\exp\left(-\beta_{1,-0}^\top \tilde{x}_{-0}\right)\right\}
    \left\{1+\exp(-\gamma_{1,0})\exp\left(-\gamma_{1,-0}^\top \tilde{x}_{-0}\right)\right\}
    \\
    &\quad=
    \left\{1+\exp(-\beta_{2,0})\exp\left(-\beta_{2,-0}^\top \tilde{x}_{-0}\right)\right\}
    \left\{1+\exp(-\gamma_{2,0})\exp\left(-\gamma_{2,-0}^\top \tilde{x}_{-0}\right)\right\}
    ,
\end{align*}
for all $\tilde{x}_{-0}\in\mathcal{U}$.
Therefore, we have
\begin{align}
    \label{eq:expansion_ident}
    &
    \exp(-\beta_{1,0})\exp(-\beta_{1,-0}^\top \tilde{x}_{-0})
    +\exp(-\gamma_{1,0})\exp(-\gamma_{1,-0}^\top \tilde{x}_{-0})
    \nonumber
    \\
    &\quad
    +\exp(-\beta_{1,0}-\gamma_{1,0})\exp(-(\beta_{1,-0}+\gamma_{1,-0})^\top \tilde{x}_{-0})
    \nonumber
    \\
    &=
    \exp(-\beta_{2,0})\exp(-\beta_{2,-0}^\top \tilde{x}_{-0})
    +\exp(-\gamma_{2,0})\exp(-\gamma_{2,-0}^\top \tilde{x}_{-0})
    \nonumber
    \\
    &\quad
    +\exp(-\beta_{2,0}-\gamma_{2,0})\exp(-(\beta_{2,-0}+\gamma_{2,-0})^\top \tilde{x}_{-0})
\end{align}
for all $\tilde{x}_{-0} \in \mathcal{U}$.
Since $\beta_{1,-0},~ \gamma_{1,-0},~ \beta_{1,-0} + \gamma_{1,-0}$ are pairwise distinct, the left-hand side of \eqref{eq:expansion_ident} is a linear combination of three distinct exponential functions and can be written as
\begin{align*}
    \sum_{j=1}^{3} \exp(c_{1,j}) \exp\left(-\lambda_{1,j}^\top \tilde{x}_{-0}\right)
    ,
\end{align*}
where $\{\lambda_{1,1},~\lambda_{1,2},~\lambda_{1,3}\} = \{\beta_{1,-0},~ \gamma_{1,-0},~ \beta_{1,-0} + \gamma_{1,-0}\}$ and $\{c_{1,1},~c_{1,2},~c_{1,3}\} = \{-\beta_{1,0},~ -\gamma_{1,0},~ -\beta_{1,0}-\gamma_{1,0}\}$.
The right-hand side can be written as
\begin{align*}
    \sum_{j=1}^{m} \exp({c_{2,j}}) \exp\left(-\lambda_{2,j}^\top \tilde{x}_{-0}\right)
    ,
    \quad
    \text{for}
    \quad
    1 \leq m \leq 3
    ,
\end{align*}
where $\lambda_{2,1},\ldots,\lambda_{2,m}$ are distinct.
Therefore, we have
\begin{align*}
    \sum_{j=1}^{3} \exp(c_{1,j}) \exp\left(-\lambda_{1,j}^\top \tilde{x}_{-0}\right) - \sum_{j=1}^{m} \exp({c_{2,j}}) \exp\left(-\lambda_{2,j}^\top \tilde{x}_{-0}\right)
    =
    0
    .
\end{align*}
By Lemma~\ref{lem:exp_linear_indep}, if the sets of vectors $\{\lambda_{1,1},~ \lambda_{1,2},~ \lambda_{1,3}\}$ and $\{\lambda_{2,1},~ \ldots,~ \lambda_{2,m}\}$ are not identical, it implies that at least one coefficient in the combined linear combination, which is either $\exp(c_{1,j})$ or $-\exp(c_{2,k})$, must be zero. 
However, this is a contradiction because the exponential function is strictly positive.
Therefore, we obtain $m=3$ and
\begin{align*}
    \left\{
        \beta_{1,-0},
        \gamma_{1,-0},
        \beta_{1,-0}+\gamma_{1,-0}
    \right\}
    =
    \left\{
        \beta_{2,-0},
        \gamma_{2,-0},
        \beta_{2,-0}+\gamma_{2,-0}
    \right\}
    .
\end{align*}

Since $\beta_{1,-0},~ \gamma_{1,-0}$, and $\beta_{1,-0}+\gamma_{1,-0}$ are pairwise distinct\footnote{This condition implies that both $\beta_{1,-0}$ and $\gamma_{1,-0}$ are non-zero, and ensures that neither can be expressed as the sum of the other two elements in the set.}, $\beta_{1,-0}+\gamma_{1,-0}$ is the unique element in the set $\{\beta_{1,-0},~ \gamma_{1,-0},~ \beta_{1,-0}+\gamma_{1,-0}\}$ that can be expressed as the sum of the other two elements.
Because the sets $\{\beta_{1,-0},~ \gamma_{1,-0},~ \beta_{1,-0}+\gamma_{1,-0}\}$ and $\{\beta_{2,-0},~ \gamma_{2,-0},~ \beta_{2,-0}+\gamma_{2,-0}\}$ are identical, their unique sum elements must be equal.
Therefore, we have
\begin{align*}
    \beta_{1,-0}+\gamma_{1,-0}
    =
    \beta_{2,-0}+\gamma_{2,-0}
    ,
\end{align*}
which implies that the sets of the remaining elements are also identical:
\begin{align*}
    \{\beta_{1,-0},~ \gamma_{1,-0}\}
    =
    \{\beta_{2,-0},~ \gamma_{2,-0}\}
    .
\end{align*}
If $\beta_{1,-0}=\beta_{2,-0}$ and $\gamma_{1,-0}=\gamma_{2,-0}$, then we have $\exp(-\beta_{1,0})=\exp(-\beta_{2,0})$ and $\exp(-\gamma_{1,0})=\exp(-\gamma_{2,0})$, and hence $\beta_1=\beta_2$ and $\gamma_1=\gamma_2$.
Instead, if $\beta_{1,-0}=\gamma_{2,-0}$ and $\gamma_{1,-0}=\beta_{2,-0}$, then we have $\exp(-\beta_{1,0})=\exp(-\gamma_{2,0})$ and $\exp(-\gamma_{1,0})=\exp(-\beta_{2,0})$, and hence $\beta_1=\gamma_2$ and $\gamma_1=\beta_2$.
Therefore, we have either $(\beta_1,~ \gamma_1)=(\beta_2,~ \gamma_2)$ or $(\beta_1,~ \gamma_1)=(\gamma_2,~ \beta_2)$.
\end{proof}

We next prove Theorem~\ref{thm:two_mode}.

\begin{proof}[Proof of Theorem~\ref{thm:two_mode}]
Fix $\xi \in \mathcal{S}$.
For $j=1,2$, define $\beta_{j,0}^{(\xi)} := \beta_{j,0}+\beta_{j,-0}^{(2)\top}\xi$ and $\gamma_{j,0}^{(\xi)} := \gamma_{j,0}+\gamma_{j,-0}^{(2)\top}\xi$.
Then, for all $y \in \{0,1\}$ and all $\tilde{x}_{-0}^{(1)} \in \mathcal{U}_{\xi}$, \eqref{eq:equiv_prob} can be expressed as
\begin{align*}
    &p\left(
        y
        \mid
        (1,\tilde{x}_{-0}^{(1)\top})^\top,
        (\beta_{1,0}^{(\xi)},\beta_{1,-0}^{(1)\top})^\top,
        (\gamma_{1,0}^{(\xi)},\gamma_{1,-0}^{(1)\top})^\top
    \right)
    \\
    &\quad
    =
    p\left(
        y
        \mid
        (1,\tilde{x}_{-0}^{(1)\top})^\top,
        (\beta_{2,0}^{(\xi)},\beta_{2,-0}^{(1)\top})^\top,
        (\gamma_{2,0}^{(\xi)},\gamma_{2,-0}^{(1)\top})^\top
    \right)
    .
\end{align*}
Because $\mathcal{U}_{\xi}$ is a nonempty open subset of $\mathbb{R}^{r}$ and $\beta_{1,-0}^{(1)},~ \gamma_{1,-0}^{(1)},~ \beta_{1,-0}^{(1)}+\gamma_{1,-0}^{(1)}$ are pairwise distinct, Proposition~\ref{prop:two_mode_c1} yields that, for each $\xi \in \mathcal{S}$, we have either
\begin{align}
    \label{eq:direct_xi}
    (\beta_{1,-0}^{(1)},~\gamma_{1,-0}^{(1)})=(\beta_{2,-0}^{(1)},~\gamma_{2,-0}^{(1)})
    \quad
    \text{and}
    \quad
    (\beta_{1,0}^{(\xi)},~\gamma_{1,0}^{(\xi)})=(\beta_{2,0}^{(\xi)},~\gamma_{2,0}^{(\xi)})
    ,
\end{align}
or
\begin{align}
    \label{eq:swap_xi}
    (\beta_{1,-0}^{(1)},~\gamma_{1,-0}^{(1)})=(\gamma_{2,-0}^{(1)},~\beta_{2,-0}^{(1)})
    \quad
    \text{and}
    \quad
    (\beta_{1,0}^{(\xi)},~\gamma_{1,0}^{(\xi)})=(\gamma_{2,0}^{(\xi)},~\beta_{2,0}^{(\xi)})
    .
\end{align}
We claim that the same set of equations must hold for all $\xi \in \mathcal{S}$.
Indeed, if \eqref{eq:direct_xi} holds for some $\xi \in \mathcal{S}$ and \eqref{eq:swap_xi} holds for some $\xi' \in \mathcal{S}$, then $\beta_{1,-0}^{(1)} = \beta_{2,-0}^{(1)} = \gamma_{1,-0}^{(1)}$, which contradicts $\beta_{1,-0}^{(1)} \neq \gamma_{1,-0}^{(1)}$.
Therefore, either \eqref{eq:direct_xi} holds for all $\xi \in \mathcal{S}$, or \eqref{eq:swap_xi} holds for all $\xi \in \mathcal{S}$.

First, suppose that \eqref{eq:direct_xi} holds for all $\xi \in \mathcal{S}$.
Then, we have
\begin{align*}
    \beta_{1,0}+\beta_{1,-0}^{(2)\top}\xi
    &=
    \beta_{2,0}+\beta_{2,-0}^{(2)\top}\xi
    ,
    \\
    \quad
    \gamma_{1,0}+\gamma_{1,-0}^{(2)\top}\xi
    &=
    \gamma_{2,0}+\gamma_{2,-0}^{(2)\top}\xi
    ,
\end{align*}
for all $\xi \in \mathcal{S}$.
Hence, both of the following functions:
\begin{align*}
    \xi
    &\mapsto
    \beta_{1,0}+\beta_{1,-0}^{(2)\top}\xi
    -
    \beta_{2,0}-\beta_{2,-0}^{(2)\top}\xi
    ,
    \\
    \xi
    &\mapsto
    \gamma_{1,0}+\gamma_{1,-0}^{(2)\top}\xi
    -
    \gamma_{2,0}-\gamma_{2,-0}^{(2)\top}\xi
    ,
\end{align*}
are identically zero on $\mathcal{S}$.
Since $\operatorname{aff}(\mathcal{S})=\mathbb{R}^{s}$, both functions are also identically zero on $\mathbb{R}^{s}$.
Therefore, we have
$\beta_{1,0}=\beta_{2,0}$, $\beta_{1,-0}^{(2)}=\beta_{2,-0}^{(2)}$, $\gamma_{1,0}=\gamma_{2,0}$, and $\gamma_{1,-0}^{(2)}=\gamma_{2,-0}^{(2)}$.
Combining \eqref{eq:direct_xi} with this, we obtain $(\beta_1,\gamma_1)=(\beta_2,\gamma_2)$.

Next, suppose that \eqref{eq:swap_xi} holds for all $\xi \in \mathcal{S}$.
Then, by a similar argument, we obtain $(\beta_1,\gamma_1)=(\gamma_2,\beta_2)$.

Therefore, we have either $(\beta_1,\gamma_1)=(\beta_2,\gamma_2)$ or $(\beta_1,\gamma_1)=(\gamma_2,\beta_2)$.
\end{proof}

We finally prove Corollary~\ref{cor:quotient_mle}.

\begin{proof}[Proof of Corollary~\ref{cor:quotient_mle}.]
Let $\pi_{\beta,\gamma}(x):=F(x^\top\beta)F(x^\top\gamma)$ and $\pi^\ast(x):=\pi_{\beta^\ast,\gamma^\ast}(x)$.
Under correct specification, we have
\begin{align*}
    y\mid x
    \sim
    \operatorname{Bernoulli}\!\left(\pi^\ast(x)\right)
    .
\end{align*}
Therefore, the expected log-likelihood function has the following expression:
\begin{align*}
    \mathcal{L}(\beta,~ \gamma)
    &=
    \mathbb{E}_{x}\left[
        \pi^\ast(x)\log \pi_{\beta,\gamma}(x)
        +\{1-\pi^\ast(x)\}\log\{1-\pi_{\beta,\gamma}(x)\}
    \right]
    ,
\end{align*}
and, hence, we have
\begin{align*}
    &\mathcal{L}(\beta,~ \gamma)-\mathcal{L}(\beta^\ast,~ \gamma^\ast)
    \\
    &\quad=
    \mathbb{E}_{x}\left[
        \pi^\ast(x)\log \pi_{\beta,\gamma}(x)
        +\{1-\pi^\ast(x)\}\log\{1-\pi_{\beta,\gamma}(x)\}
    \right]
    \\
    &\qquad
    -
    \mathbb{E}_{x}\left[
        \pi^\ast(x)\log \pi^\ast(x)
        +\{1-\pi^\ast(x)\}\log\{1-\pi^\ast(x)\}
    \right]
    \\
    &\quad=
    -\mathbb{E}_{x}\left[
        \pi^\ast(x)\log\frac{\pi^\ast(x)}{\pi_{\beta,\gamma}(x)}
        +\{1-\pi^\ast(x)\}\log\frac{1-\pi^\ast(x)}{1-\pi_{\beta,\gamma}(x)}
    \right]
    \\
    &\quad=
    -\mathbb{E}_{x}\left[
        \mathrm{KL}\left(
            \operatorname{Bernoulli}~\left(\pi^\ast(x)\right)
            ~\middle\|~
            \operatorname{Bernoulli}~\left(\pi_{\beta,\gamma}(x)\right)
        \right)
    \right]
    \leq
    0
    ,
\end{align*}
for any $(\beta,\gamma) \in \mathbb{R}^d\times\mathbb{R}^d$.
Here, $\mathrm{KL}(\cdot \| \cdot)$ denotes the Kullback-Leibler divergence from the first probability distribution to the second.
Thus $[\beta^\ast,~\gamma^\ast]$ is a maximizer on $(\mathbb{R}^d\times\mathbb{R}^d)/{\sim}$.

Now suppose that $(\beta^{\dagger},\gamma^{\dagger}) \in \mathbb{R}^d\times\mathbb{R}^d$ also attains the maximum.
Then, it must hold that
\begin{align*}
    \mathrm{KL}\left(
        \operatorname{Bernoulli}~\left(\pi^\ast(x)\right)
        ~\middle\|~
        \operatorname{Bernoulli}~\left(\pi_{\beta^{\dagger},\gamma^{\dagger}}(x)\right)
    \right)
    =
    0
    ,
    \quad
    \text{almost surely.}
\end{align*}
Therefore, we obtain
\begin{align}
    \label{eq:equality_pi}
    \pi_{\beta^{\dagger},\gamma^{\dagger}}(x)
    =
    \pi^\ast(x)
    ,
    \quad
    \text{for almost every $x$.}
\end{align}
Let $x=(1,\tilde{x}_{-0}^{(1)\top},\tilde{x}_{-0}^{(2)\top})^\top$.
Fix $\xi \in \mathcal{S}$ and define
\begin{align*}
    g_{\dagger,\xi}(\tilde{x}_{-0}^{(1)})
    &:=
    \pi_{\beta^{\dagger},\gamma^{\dagger}} \left((1,\tilde{x}_{-0}^{(1)\top},\xi^\top)^\top\right)
    ,
    \\
    g_{\ast,\xi}(\tilde{x}_{-0}^{(1)})
    &:=
    \pi^\ast \left((1,\tilde{x}_{-0}^{(1)\top},\xi^\top)^\top\right)
    .
\end{align*}
Both functions are continuous on $\mathbb{R}^{r}$.
We claim that $g_{\dagger,\xi}(\tilde{x}_{-0}^{(1)}) = g_{\ast,\xi}(\tilde{x}_{-0}^{(1)})$ holds for all $\xi \in \mathcal{S}$ and all $\tilde{x}_{-0}^{(1)} \in \mathcal{U}_{\xi}$.
Suppose not, then there exist $\xi \in \mathcal{S}$ and $w \in \mathcal{U}_{\xi}$ such that $g_{\dagger,\xi}(w) \neq g_{\ast,\xi}(w)$.
By continuity, there exists a nonempty open neighborhood $V \subset \mathcal{U}_{\xi}$ of $w$ such that the two functions remain different on $V$.
Because $\xi \in \operatorname{supp}(\tilde{x}_{-0}^{(2)})$ and $w$ belongs to the conditional support of $\tilde{x}_{-0}^{(1)}$ given $\tilde{x}_{-0}^{(2)}=\xi$, we have $\mathbb{P}(\tilde{x}_{-0}^{(2)}=\xi,~\tilde{x}_{-0}^{(1)}\in V) > 0$.
This contradicts \eqref{eq:equality_pi}.
Therefore, we obtain
\begin{align*}
    \pi_{\beta^{\dagger},\gamma^{\dagger}}(x)
    =
    \pi^\ast(x)
\end{align*}
for all $x=(1,\tilde{x}_{-0}^{(1)\top},\xi^\top)^\top$ with $\xi \in \mathcal{S}$ and $\tilde{x}_{-0}^{(1)}\in\mathcal{U}_{\xi}$.

From the discussion above, we have
\begin{align*}
    p(y \mid x, \beta^{\dagger}, \gamma^{\dagger})
    =
    p(y \mid x, \beta^\ast, \gamma^\ast)
    ,
\end{align*}
for all $y\in\{0,1\}$ and all $x=(1,\tilde{x}_{-0}^{(1)\top},\xi^\top)^\top$ with $\xi \in \mathcal{S}$ and $\tilde{x}_{-0}^{(1)}\in\mathcal{U}_{\xi}$.
By Theorem~\ref{thm:two_mode}, it follows that
\begin{align*}
    (\beta^{\dagger},~ \gamma^{\dagger})
    \sim
    (\beta^\ast,~ \gamma^\ast)
    .
\end{align*}
Therefore, the expected log-likelihood $\mathcal{L}(\beta,~ \gamma)$ is uniquely maximized on $(\mathbb{R}^d\times\mathbb{R}^d)/{\sim}$ at the class $[\beta^\ast,~ \gamma^\ast]$.
\end{proof}

\section{Details of Numerical Settings of Bimodality Confirmation}
\label{apx:numerical_details}

\subsection{Posterior Distribution and Sampling Algorithm}

We consider the posterior distribution given by
\begin{equation}
    \label{eq:tempered_posterior}
    \pi(\beta,\gamma,h\mid y,x,z)
    \propto
    \left\{\prod_{i=1}^n p(y_i,h_i\mid x_i,z_i,\beta,\gamma)\right\}
    p_{\mathrm{prior}}(\beta,\gamma)
    ,
\end{equation}
where $p(y_i,h_i\mid x_i,z_i,\beta,\gamma)$ is defined as
\begin{align*}
    p(y_i,h_i\mid x_i,z_i,\beta,\gamma)
    &=
    \left\{F(\gamma^\top z_i)\right\}^{h_i}
    \left\{1-F(\gamma^\top z_i)\right\}^{1-h_i}
    \left\{F(\beta^\top x_i)\right\}^{y_i}
    \left\{1-F(\beta^\top x_i)\right\}^{h_i-y_i}
    ,
\end{align*}
and $p_{\mathrm{prior}}(\beta,\gamma)$ denotes the prior distribution.
The complete-data likelihood can be separated into a logistic regression for $h$:
\begin{align*}
    h_i\mid z_i
    \sim
    \operatorname{Bernoulli}(F(\gamma^\top z_i))
    .
\end{align*}
and a logistic regression for $y^\ast$:
\begin{align*}
    y_i^\ast \mid x_i
    \sim
    \operatorname{Bernoulli}(F(\beta^\top x_i))
    ,
\end{align*}
conditioned on $h_i=1$.
Thus $h_i=0$ generates a structural zero, whereas $h_i=1$ allows the ordinary logistic regression.
The observed outcome $y_i$ is then given by $y_i = h_i y_i^\ast$.
We use this structure to construct a MCMC algorithm.
We further employ the P\'olya-Gamma augmentation \citep{polson2013bayesian} and combine the resulting Gibbs sampling algorithm with replica exchange so that the sampler can move between the multiple modes.
The detailed algorithms are provided in Appendix~\ref{apx:sampling_algorithm}.

\subsection{Data Generation and Sampling Setup}

Numerical experiments were conducted under the shared-design setting.
We considered three designs that differ only in the distribution of the covariates.
In each scenario, we generated $n=2,000$ observations with $d=p=5$ covariates including an intercept.
The true coefficient vectors were fixed at
\begin{align*}
    \beta^{*} &= (0.5,~1.0,~ 0.5,~ 0.5,~ 0.25)^\top
    ,
    \\
    \gamma^{*} &= (1.7,~-1.0,~ -1.0,~ 0.5,~ 0.5)^\top
    .
\end{align*}
In Scenario~1, all non-intercept covariates were drawn from independent standard normal distributions.
In Scenario~2, all non-intercept covariates were drawn from independent $\operatorname{Bernoulli}(0.5)$.
In Scenario~3, the first two non-intercept covariates were drawn from independent standard normal distributions, and the remaining ones were from independent $\operatorname{Bernoulli}(0.5)$.
Given the covariates, we generated
\begin{align*}
    h_i
    &\sim
    \operatorname{Bernoulli}\left(F(\gamma^{* \top} x_i)\right)
    ,
    \\
    y_i^\ast
    &\sim
    \operatorname{Bernoulli}\left(F(\beta^{* \top} x_i)\right)
    ,
\end{align*}
independently, and set $y_i = h_i y_i^\ast$.

We placed weakly informative Gaussian priors on both coefficient vectors:
$\beta\sim\mathcal{N}(0,100I_p)$ and $\gamma\sim\mathcal{N}(0,100I_q)$.
Posterior sampling was performed via a Gibbs sampling algorithm with replica exchange using $20$ replicas.
The temperature schedule followed a geometric progression $T_m = r^m$ with $r = 1.05$, and replica exchange was attempted every $50$ iterations. 
The total number of MCMC iterations was set to $53{,}000$, with the first $3{,}000$ discarded as burn-in.
To facilitate efficient sampling, Gibbs sampling based on P\'olya-Gamma data augmentation was used for both the ordinary logistic regression component and the structural-zero component.

\subsection{Sampling Results}

To explore the structure of the posterior distribution, we applied $k$-means++ clustering algorithm with $k=2$ to the posterior samples \cite{arthur2007kmeans}.
The samples, consisting of both $\beta$ and $\gamma$ parameters concatenated as $10$-dimensional vectors, were projected onto the first two principal components for visualization using PCA.
See Figure~\ref{fig:posterior_pca} for the plots.

Table~\ref{tab:posterior_means} summarizes the posterior means within each cluster along with cluster sizes and proportions.
In Scenarios~1 and 3, the means of the two clusters exhibited an approximately symmetric structure, reflecting the exchange symmetry of $(\beta, \gamma)$ and $(\gamma, \beta)$.
In Scenario~2, the posterior means from the smaller cluster were numerically large, especially in the structural-zero component, consistent with the failure of the binary design to satisfy the condition (C2), which results in a lack of guaranteed identifiability.

The trace plots and the histograms for the posterior distributions are provided in Appendix~\ref{apx:trace_hist_posterior}.

\begin{table}[htbp]
    \centering
    \caption{Posterior means of each parameter for the clusters in Scenarios 1--3.}
    \label{tab:posterior_means}
    \begin{tabular}{cccccccc}
        \toprule
        \multirow{2}{*}{Parameter}& \multirow{2}{*}{True} & \multicolumn{2}{c}{Scenario 1} & \multicolumn{2}{c}{Scenario 2} & \multicolumn{2}{c}{Scenario 3} \\
        \cmidrule(lr){3-4} \cmidrule(lr){5-6} \cmidrule(lr){7-8}
        & {} & Cluster 1 & Cluster 2 & Cluster 1 & Cluster 2 & Cluster 1 & Cluster 2 \\
        \midrule
        $\beta_0$ & 0.5 & 0.325 & 2.886 & 9.700 & 17.168 & 0.202 & 2.713 \\
        $\beta_1$ & 1.0 & 0.871 & -1.349 & 5.136 & 0.847 & 0.760 & -1.453 \\
        $\beta_2$ & 0.5 & 0.355 & -1.572 & 3.180 & 5.330 & 0.234 & -1.283 \\
        $\beta_3$ & 0.5 & 0.515 & 0.544 & 4.194 & -0.704 & 0.532 & 0.737 \\
        $\beta_4$ & 0.25 & 0.282 & 0.668 & 8.992 & -0.187 & 0.574 & 0.158 \\
        $\gamma_0$ & 1.7 & 2.413 & 0.200 & 0.280 & 0.265 & 2.794 & 0.216 \\
        $\gamma_1$ & -1.0 & -1.154 & 0.803 & -0.214 & -0.204 & -1.491 & 0.773 \\
        $\gamma_2$ & -1.0 & -1.413 & 0.284 & -0.622 & -0.622 & -1.307 & 0.249 \\
        $\gamma_3$ & 0.5 & 0.524 & 0.516 & 0.570 & 0.575 & 0.749 & 0.532 \\
        $\gamma_4$ & 0.5 & 0.625 & 0.292 & 0.470 & 0.485 & 0.154 & 0.575 \\
        \midrule
        \multicolumn{2}{c}{Cluster size} & 13,250 & 36,750 & 8,395 & 41,605 & 25,450 & 24,550 \\
        \multicolumn{2}{c}{Proportion} & 0.265 & 0.735 & 0.168 & 0.832 & 0.509 & 0.491 \\
        \bottomrule
    \end{tabular}
\end{table}

\section{Details of Sampling Algorithm}
\label{apx:sampling_algorithm}

For each replica $m=1,\ldots,M$, let $T_m>0$ denote the temperature.
The tempered posterior distribution is defined as
\begin{align*}
    \pi_{T_m}(\beta,\gamma,h\mid y,x,z)
    \propto
    \left\{\prod_{i=1}^n p(y_i,h_i\mid x_i,z_i,\beta,\gamma)\right\}^{1/T_m}
    p_{\mathrm{prior}}(\beta,\gamma)
    ,
\end{align*}
where
\begin{align*}
    p(y_i,h_i\mid x_i,z_i,\beta,\gamma)
    =
    \left\{F(z_i^\top\gamma)\right\}^{h_i}
    \left\{1-F(z_i^\top\gamma)\right\}^{1-h_i}
    \left\{F(x_i^\top\beta)\right\}^{y_i}
    \left\{1-F(x_i^\top\beta)\right\}^{h_i-y_i}
    ,
\end{align*}
for $y_i\leq h_i$.
We assume Gaussian priors: $\beta\sim\mathcal{N}(b_0,B_0)$, and $\gamma\sim\mathcal{N}(g_0,G_0)$.

\subsection{P\'olya-Gamma Gibbs Sampling Step}

For each replica $m$, we perform Gibbs sampling using the following steps.

\paragraph{Step 1. Updating $h_i$.}

If $y_i = 1$, then $h_i$ is deterministically set to $1$.  
For observations with $y_i = 0$, we sample $h_i$ from
\begin{align*}
    h_i\mid \beta,\gamma,y_i=0
    \sim
    \operatorname{Bernoulli} (\mu_i^{(T_m)})
    ,
\end{align*}
where
\begin{align*}
    \mu_i^{(T_m)}
    &=
    \frac{\left[F(z_i^\top\gamma)\{1-F(x_i^\top\beta)\}\right]^{1/T_m}}
       {\left[F(z_i^\top\gamma)\{1-F(x_i^\top\beta)\}\right]^{1/T_m}
        +\left[1-F(z_i^\top\gamma)\right]^{1/T_m}}
    \\
    &=
    F\!\left(\frac{z_i^\top\gamma-\log\{1+\exp(x_i^\top\beta)\}}{T_m}\right)
    .
\end{align*}

\paragraph{Step 2. Updating $\gamma$.}

The tempered likelihood for $\gamma$ is
\begin{align*}
    \prod_{i=1}^n
    \frac{\exp \left\{(h_i/T_m) z_i^\top\gamma\right\}}
     {\{1+\exp(z_i^\top\gamma)\}^{1/T_m}}
     .
\end{align*}
Introduce P\'olya-Gamma auxiliary variables \citep{polson2013bayesian}:
\begin{align*}
    w_i \mid \gamma, h \sim \operatorname{PG}(T_m^{-1}, z_i^\top\gamma)
    ,
    \quad
    i=1,\ldots,n
    ,
\end{align*}
where $\operatorname{PG}(\cdot,\cdot)$ denotes the P\'olya-Gamma distribution.
Specifically, if $w \sim \operatorname{PG}(b, c)$ with $b>0$ and $c \in \mathbb{R}$, using random numbers independently following Gamma distribution $v_k \sim \operatorname{Ga}(b, 1)$, $w$ can be obtained as
\begin{align*}
    w
    =
    \cfrac{1}{2 \pi^2} \sum_{k=1}^{\infty} \cfrac{v_k}{\left(k-\frac{1}{2}\right)^2+\frac{c^2}{4 \pi^2}}.
\end{align*}
Let $Z=(z_1^\top,\ldots,z_n^\top)^\top$, $W=\operatorname{diag}(w_1,\ldots,w_n)$, and
\begin{align*}
    \kappa^{(\gamma)}
    =
    \left(\frac{h_1-1/2}{T_m},~\ldots,~\frac{h_n-1/2}{T_m}\right)^\top
    .
\end{align*}
Then the full conditional distribution of $\gamma$ becomes Gaussian:
\begin{align*}
    \gamma \mid h,w,z
    &\sim \mathcal{N}\!\left(g_1^{(T_m)}, G_1^{(T_m)}\right)
    ,
    \\
    G_1^{(T_m)}
    &=
    \left(Z^\top W Z + G_0^{-1}\right)^{-1}
    ,
    \\
    g_1^{(T_m)}
    &=
    G_1^{(T_m)}
    \left(Z^\top \kappa^{(\gamma)} + G_0^{-1}g_0\right)
    .
\end{align*}

\paragraph{Step 3. Updating $\beta$.}
Let $I_h=\{i:h_i=1\}$, let $X_h$ be the submatrix of $X = (x_1, \ldots, x_n)^{\top} \in \mathbb{R}^{n \times p}$ indexed by $I_h$, and let $y_h$ be the corresponding subvector of $y$.
The tempered likelihood of $\beta$ is:
\begin{align*}
    \prod_{i\in I_h}
    \frac{\exp\left\{(y_i/T_m)x_i^\top\beta\right\}}
    {\{1+\exp(x_i^\top\beta)\}^{1/T_m}}
    .
\end{align*}
Introduce independent P\'olya-Gamma variables
\begin{align*}
    \omega_i \mid \beta,h,y \sim \operatorname{PG}\!\left(T_m^{-1}, x_i^\top\beta\right)
    ,
    \quad
    i\in I_h
    ,
\end{align*}
and, define $\Omega_h=\operatorname{diag}(\omega_i:i\in I_h)$ and
\begin{align*}
    \kappa^{(\beta)}
    =
    \left(\frac{y_i-1/2}{T_m}: i\in I_h\right)^\top
    .
\end{align*}
Then the full conditional distribution of $\beta$ becomes again Gaussian:
\begin{align*}
    \beta \mid y_h,X_h,\Omega_h,h
    &\sim \mathcal{N}\!\left(b_1^{(T_m)}, B_1^{(T_m)}\right)
    ,\\
    B_1^{(T_m)}
    &=
    \left(X_h^\top \Omega_h X_h + B_0^{-1}\right)^{-1}
    ,\\
    b_1^{(T_m)}
    &=
    B_1^{(T_m)}
    \left(X_h^\top \kappa^{(\beta)} + B_0^{-1}b_0\right)
    .
\end{align*}

\subsection{Replica Exchange Step}

After a fixed number of Gibbs iterations, we attempt to swap the states of two neighboring replicas with adjacent temperatures $T_m$ and $T_{m+1}$.
Let the current states be denoted by
\begin{align*}
    \theta^{(m)}
    &=
    (\beta^{(m)},\gamma^{(m)},h^{(m)})
    ,
    \\
    \theta^{(m+1)}
    &=
    (\beta^{(m+1)},\gamma^{(m+1)},h^{(m+1)})
    .
\end{align*}
The acceptance probability for the swap proposal is given by
\begin{equation}
  \label{eq:exchange_acceptance}
  \alpha
  =
  \min\left\{
    1,~
    \frac{\pi_{T_m}(\theta^{(m+1)}\mid y,x,z)~
          \pi_{T_{m+1}}(\theta^{(m)}\mid y,x,z)}
         {\pi_{T_m}(\theta^{(m)}\mid y,x,z)~
          \pi_{T_{m+1}}(\theta^{(m+1)}\mid y,x,z)}
  \right\}
  .
\end{equation}
In practice, we compute the log acceptance ratio as
\begin{align*}
    \log \alpha
    =
    \min\left\{
        0,~
        \left(\frac{1}{T_m}-\frac{1}{T_{m+1}}\right)
            \left[
                \log \tilde{L} \left(\theta^{(m+1)}\right) - \log \tilde{L} \left(\theta^{(m)}\right)
            \right]
    \right\}
    ,
\end{align*}
where
\begin{align*}
    \tilde{L}(\theta)
    :=
    \prod_{i=1}^n p(y_i,h_i\mid x_i,z_i,\beta,\gamma)
    .
\end{align*}

\subsection{Overall Algorithm}

The full algorithm alternates between the Gibbs sampling step and the replica exchange step as follows:
\begin{enumerate}
  \item For each replica $m$, perform Gibbs sampling to update $h$, $\gamma$, and $\beta$ from the full conditionals above under its corresponding temperature $T_m$.
  \item Every fixed number of iterations, attempt to swap the states of neighboring replicas $m$ and $m+1$ according to the acceptance probability given in \eqref{eq:exchange_acceptance}.
  \item Retain samples from the chain with $T_1=1$ as draws from the target posterior distribution.
\end{enumerate}

\section{Trace Plots and Posterior Histograms}
\label{apx:trace_hist_posterior}

We provide the trace plots for Scenario 1 as Figure~\ref{fig:trace_s1}, Scenario 2 as Figure~\ref{fig:trace_s2}, and Scenario 3 as Figure~\ref{fig:trace_s3}.
Furthermore, we provide the posterior histograms for Scenario 1 as Figure~\ref{fig:hist_s1}, Scenario 2 as Figure~\ref{fig:hist_s2}, and Scenario 3 as Figure~\ref{fig:hist_s3}.

\begin{figure}[htbp]
    \centering
    \begin{subfigure}[t]{0.40\textwidth}
        \centering
        \includegraphics[width=\linewidth]{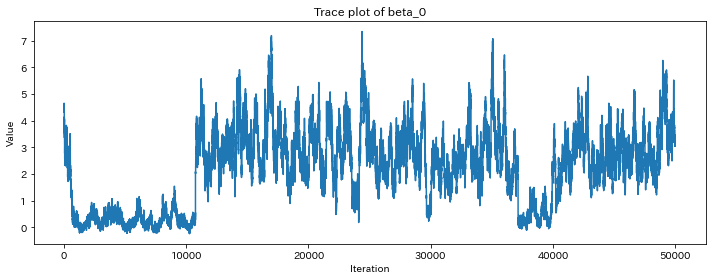}
        \caption{$\beta_0$}
    \end{subfigure}
    \hfill
    \begin{subfigure}[t]{0.40\textwidth}
        \centering
        \includegraphics[width=\linewidth]{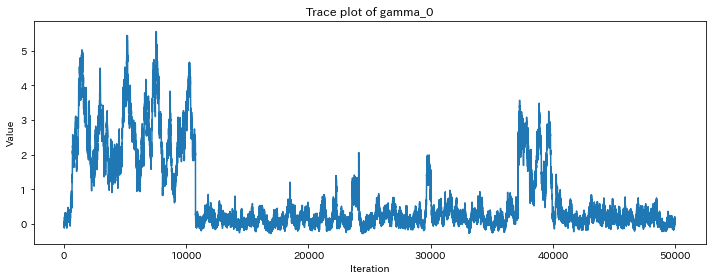}
        \caption{$\gamma_0$}
    \end{subfigure}
    \\
    \begin{subfigure}[t]{0.40\textwidth}
        \centering
        \includegraphics[width=\linewidth]{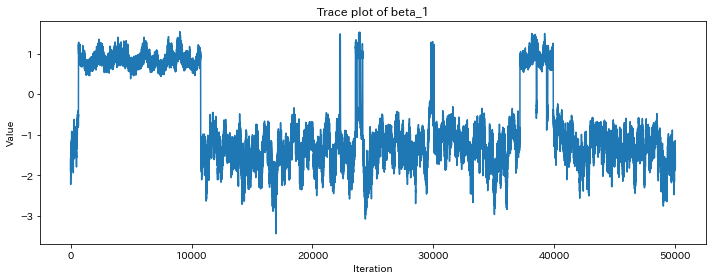}
        \caption{$\beta_1$}
    \end{subfigure}
    \hfill
    \begin{subfigure}[t]{0.40\textwidth}
        \centering
        \includegraphics[width=\linewidth]{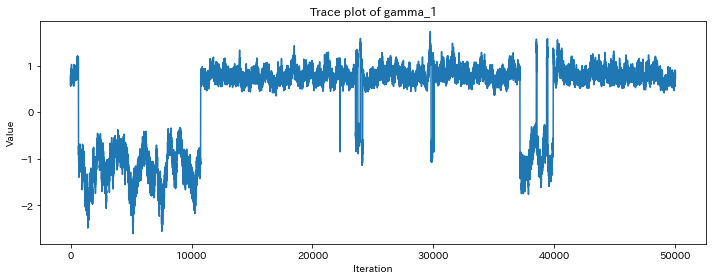}
        \caption{$\gamma_1$}
    \end{subfigure}
    \\
    \begin{subfigure}[t]{0.40\textwidth}
        \centering
        \includegraphics[width=\linewidth]{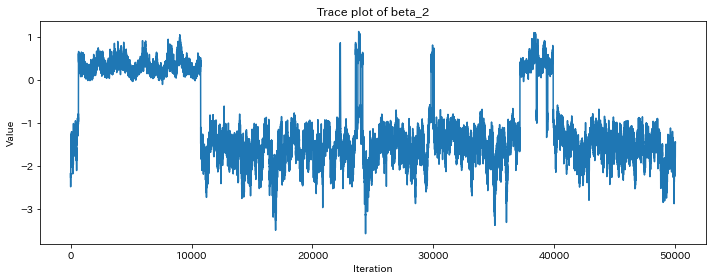}
        \caption{$\beta_2$}
    \end{subfigure}
    \hfill
    \begin{subfigure}[t]{0.40\textwidth}
        \centering
        \includegraphics[width=\linewidth]{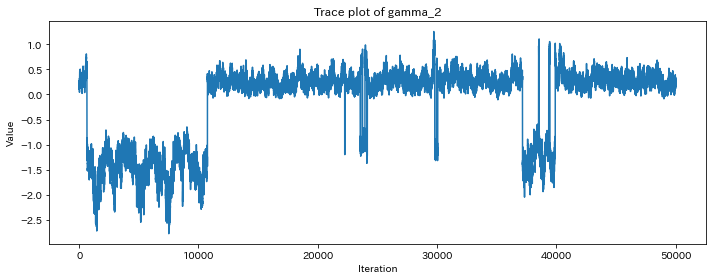}
        \caption{$\gamma_2$}
    \end{subfigure}
    \\
    \begin{subfigure}[t]{0.40\textwidth}
        \centering
        \includegraphics[width=\linewidth]{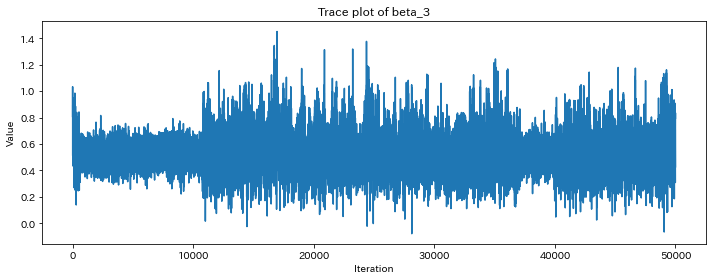}
        \caption{$\beta_3$}
    \end{subfigure}
    \hfill
    \begin{subfigure}[t]{0.40\textwidth}
        \centering
        \includegraphics[width=\linewidth]{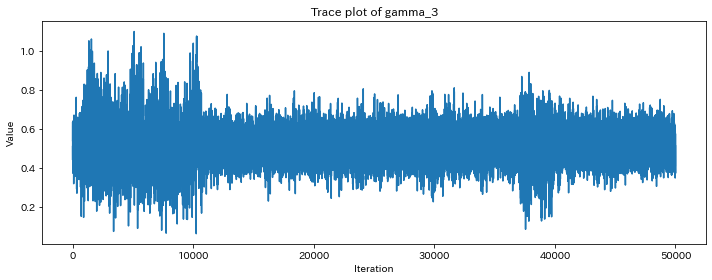}
        \caption{$\gamma_3$}
    \end{subfigure}
    \\
    \begin{subfigure}[t]{0.40\textwidth}
        \centering
        \includegraphics[width=\linewidth]{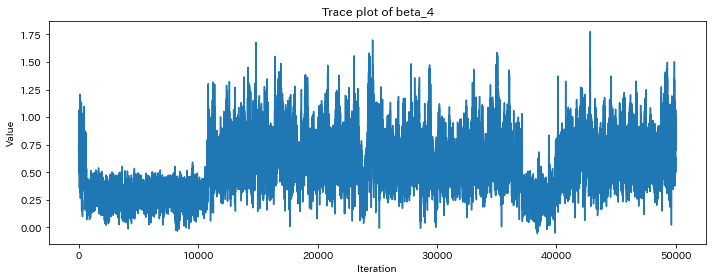}
        \caption{$\beta_4$}
    \end{subfigure}
    \hfill
    \begin{subfigure}[t]{0.40\textwidth}
        \centering
        \includegraphics[width=\linewidth]{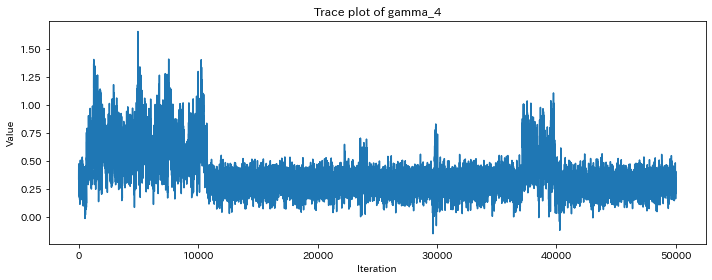}
        \caption{$\gamma_4$}
    \end{subfigure}
    \caption{Trace plots of each parameter (Scenario 1).}
    \label{fig:trace_s1}
\end{figure}

\begin{figure}[htbp]
    \centering
    \begin{subfigure}[t]{0.35\textwidth}
        \centering
        \includegraphics[width=\linewidth]{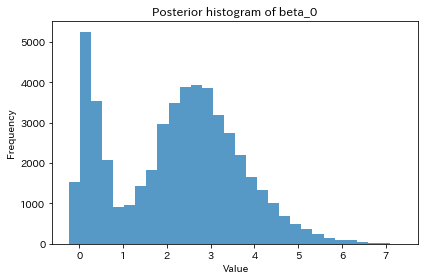}
        \caption{$\beta_0$}
    \end{subfigure}
    \hfill
    \begin{subfigure}[t]{0.35\textwidth}
        \centering
        \includegraphics[width=\linewidth]{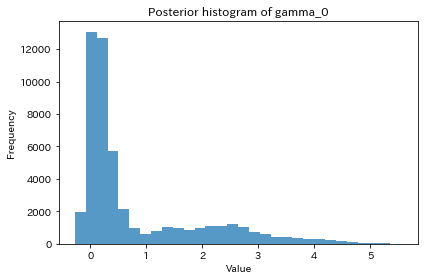}
        \caption{$\gamma_0$}
    \end{subfigure}
    \\
    \begin{subfigure}[t]{0.35\textwidth}
        \centering
        \includegraphics[width=\linewidth]{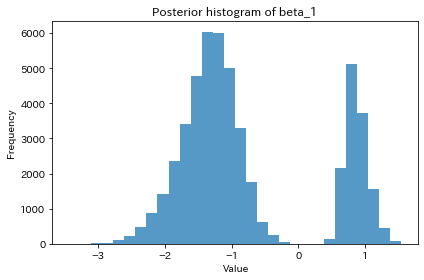}
        \caption{$\beta_1$}
    \end{subfigure}
    \hfill
    \begin{subfigure}[t]{0.35\textwidth}
        \centering
        \includegraphics[width=\linewidth]{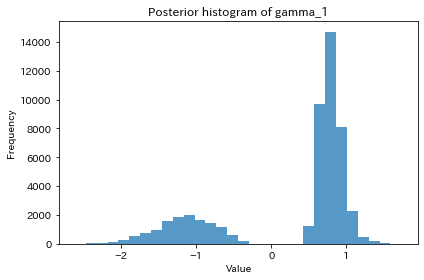}
        \caption{$\gamma_1$}
    \end{subfigure}
    \\
    \begin{subfigure}[t]{0.35\textwidth}
        \centering
        \includegraphics[width=\linewidth]{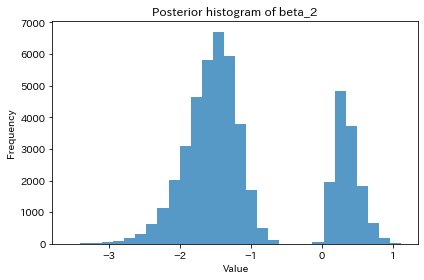}
        \caption{$\beta_2$}
    \end{subfigure}
    \hfill
    \begin{subfigure}[t]{0.35\textwidth}
        \centering
        \includegraphics[width=\linewidth]{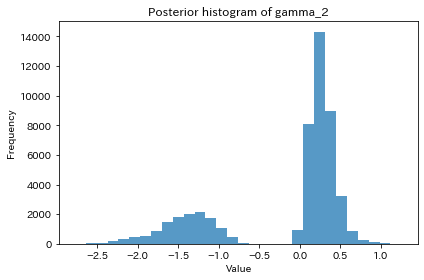}
        \caption{$\gamma_2$}
    \end{subfigure}
    \\
    \begin{subfigure}[t]{0.35\textwidth}
        \centering
        \includegraphics[width=\linewidth]{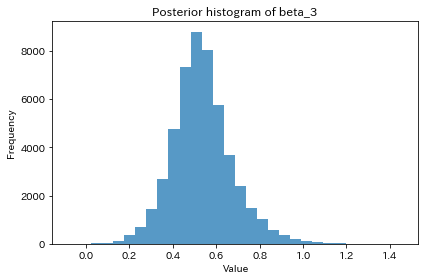}
        \caption{$\beta_3$}
    \end{subfigure}
    \hfill
    \begin{subfigure}[t]{0.35\textwidth}
        \centering
        \includegraphics[width=\linewidth]{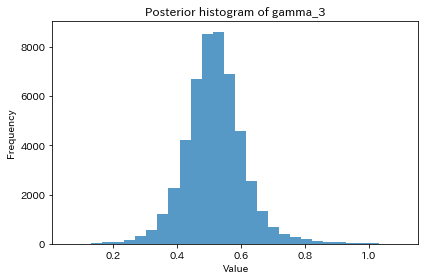}
        \caption{$\gamma_3$}
    \end{subfigure}
    \\
    \begin{subfigure}[t]{0.35\textwidth}
        \centering
        \includegraphics[width=\linewidth]{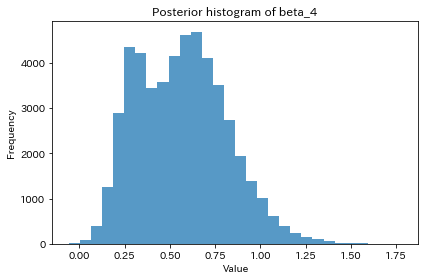}
        \caption{$\beta_4$}
    \end{subfigure}
    \hfill
    \begin{subfigure}[t]{0.35\textwidth}
        \centering
        \includegraphics[width=\linewidth]{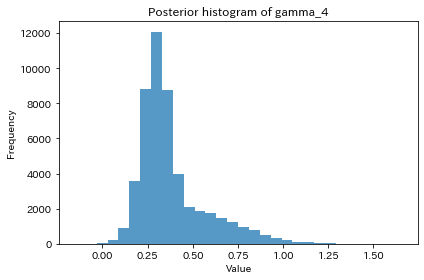}
        \caption{$\gamma_4$}
    \end{subfigure}
    \caption{The histograms of the posterior distributions for each parameter (Scenario 1).}
    \label{fig:hist_s1}
\end{figure}

\begin{figure}[htbp]
    \centering
    \begin{subfigure}[t]{0.40\textwidth}
        \centering
        \includegraphics[width=\linewidth]{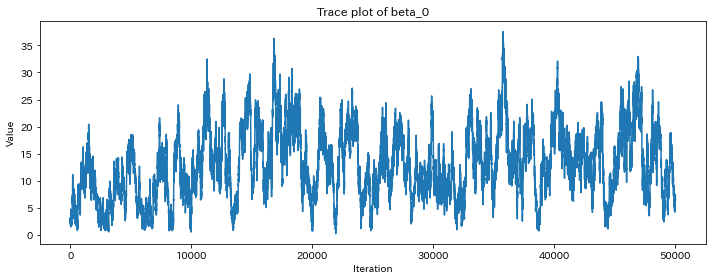}
        \caption{$\beta_0$}
    \end{subfigure}
    \hfill
    \begin{subfigure}[t]{0.40\textwidth}
        \centering
        \includegraphics[width=\linewidth]{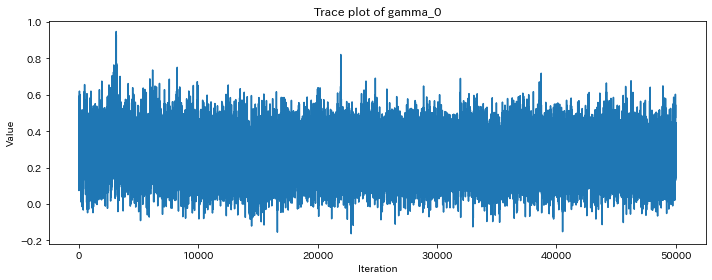}
        \caption{$\gamma_0$}
    \end{subfigure}
    \\
    \begin{subfigure}[t]{0.40\textwidth}
        \centering
        \includegraphics[width=\linewidth]{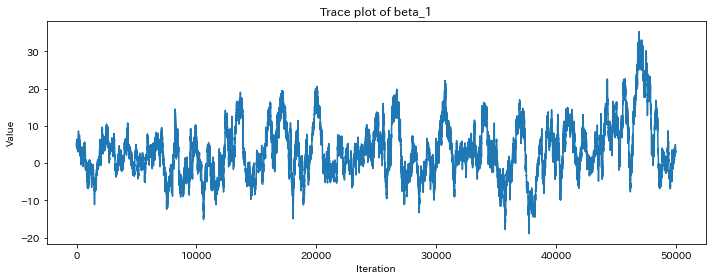}
        \caption{$\beta_1$}
    \end{subfigure}
    \hfill
    \begin{subfigure}[t]{0.40\textwidth}
        \centering
        \includegraphics[width=\linewidth]{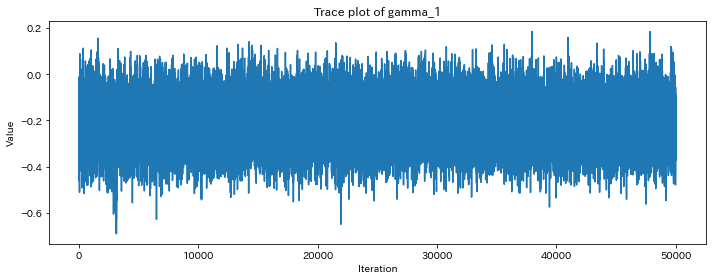}
        \caption{$\gamma_1$}
    \end{subfigure}
    \\
    \begin{subfigure}[t]{0.40\textwidth}
        \centering
        \includegraphics[width=\linewidth]{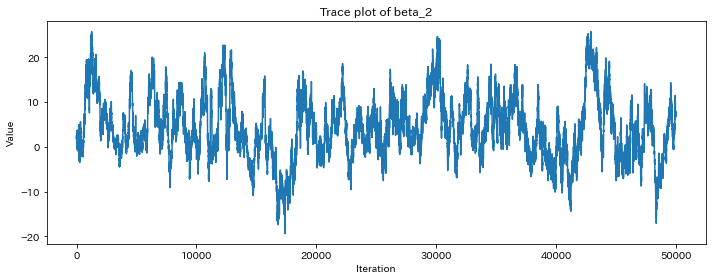}
        \caption{$\beta_2$}
    \end{subfigure}
    \hfill
    \begin{subfigure}[t]{0.40\textwidth}
        \centering
        \includegraphics[width=\linewidth]{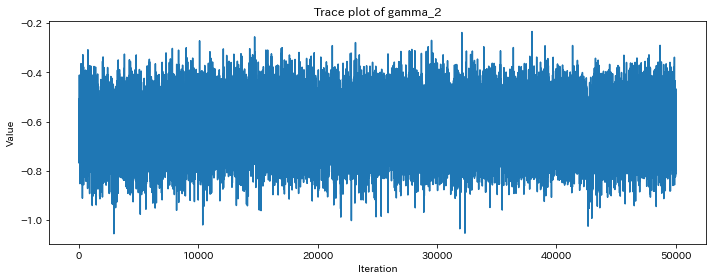}
        \caption{$\gamma_2$}
    \end{subfigure}
    \\
    \begin{subfigure}[t]{0.40\textwidth}
        \centering
        \includegraphics[width=\linewidth]{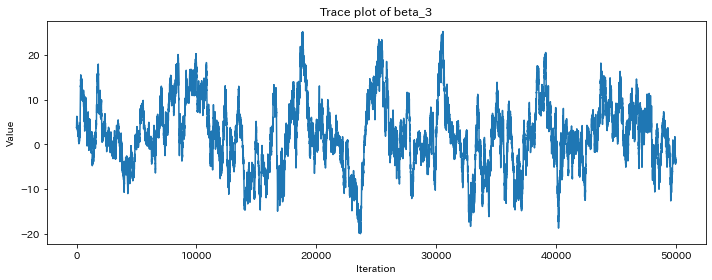}
        \caption{$\beta_3$}
    \end{subfigure}
    \hfill
    \begin{subfigure}[t]{0.40\textwidth}
        \centering
        \includegraphics[width=\linewidth]{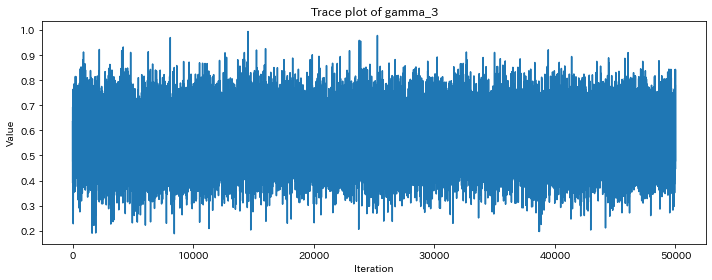}
        \caption{$\gamma_3$}
    \end{subfigure}
    \\
    \begin{subfigure}[t]{0.40\textwidth}
        \centering
        \includegraphics[width=\linewidth]{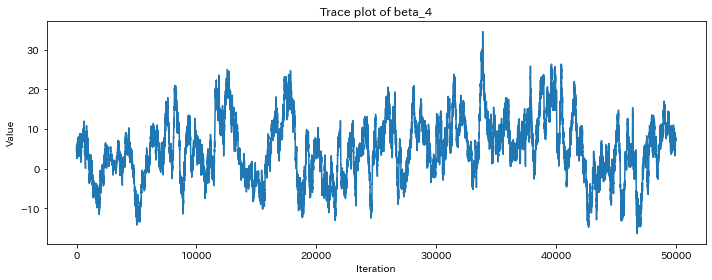}
        \caption{$\beta_4$}
    \end{subfigure}
    \hfill
    \begin{subfigure}[t]{0.40\textwidth}
        \centering
        \includegraphics[width=\linewidth]{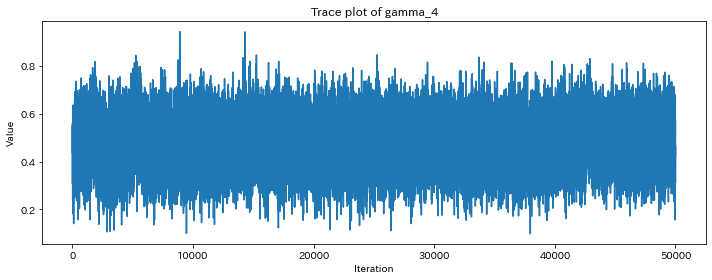}
        \caption{$\gamma_4$}
    \end{subfigure}
    \caption{Trace plots of each parameter (Scenario 2).}
    \label{fig:trace_s2}
\end{figure}

\begin{figure}[htbp]
    \centering
    \begin{subfigure}[t]{0.35\textwidth}
        \centering
        \includegraphics[width=\linewidth]{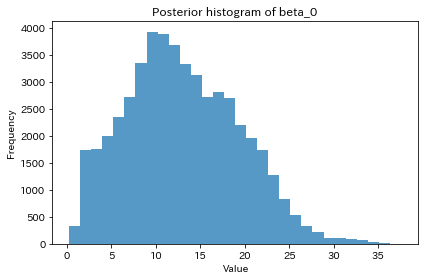}
        \caption{$\beta_0$}
    \end{subfigure}
    \hfill
    \begin{subfigure}[t]{0.35\textwidth}
        \centering
        \includegraphics[width=\linewidth]{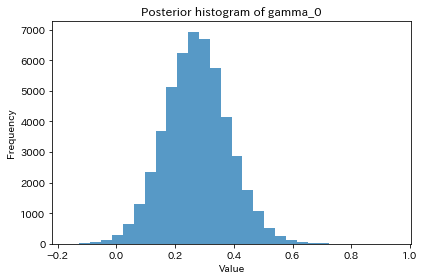}
        \caption{$\gamma_0$}
    \end{subfigure}
    \\
    \begin{subfigure}[t]{0.35\textwidth}
        \centering
        \includegraphics[width=\linewidth]{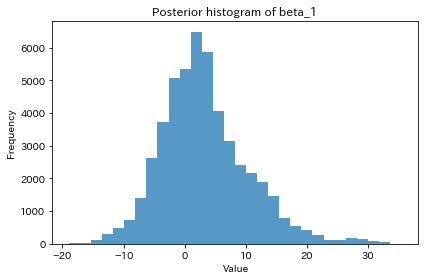}
        \caption{$\beta_1$}
    \end{subfigure}
    \hfill
    \begin{subfigure}[t]{0.35\textwidth}
        \centering
        \includegraphics[width=\linewidth]{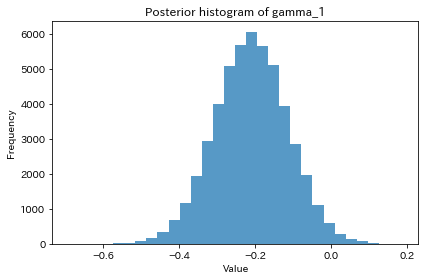}
        \caption{$\gamma_1$}
    \end{subfigure}
    \\
    \begin{subfigure}[t]{0.35\textwidth}
        \centering
        \includegraphics[width=\linewidth]{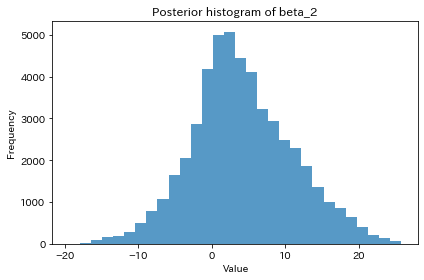}
        \caption{$\beta_2$}
    \end{subfigure}
    \hfill
    \begin{subfigure}[t]{0.35\textwidth}
        \centering
        \includegraphics[width=\linewidth]{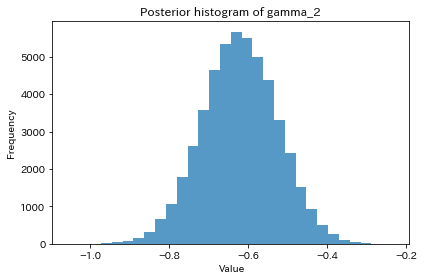}
        \caption{$\gamma_2$}
    \end{subfigure}
    \\
    \begin{subfigure}[t]{0.35\textwidth}
        \centering
        \includegraphics[width=\linewidth]{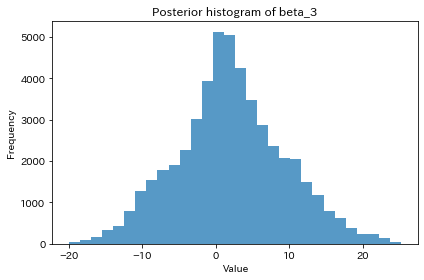}
        \caption{$\beta_3$}
    \end{subfigure}
    \hfill
    \begin{subfigure}[t]{0.35\textwidth}
        \centering
        \includegraphics[width=\linewidth]{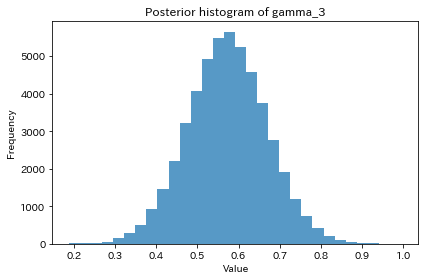}
        \caption{$\gamma_3$}
    \end{subfigure}
    \\
    \begin{subfigure}[t]{0.35\textwidth}
        \centering
        \includegraphics[width=\linewidth]{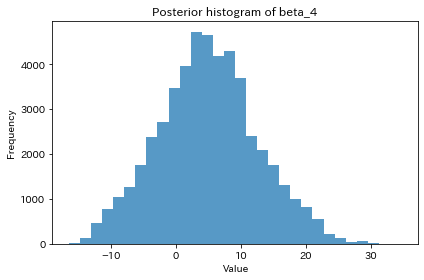}
        \caption{$\beta_4$}
    \end{subfigure}
    \hfill
    \begin{subfigure}[t]{0.35\textwidth}
        \centering
        \includegraphics[width=\linewidth]{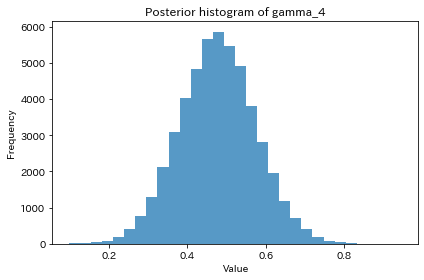}
        \caption{$\gamma_4$}
    \end{subfigure}
    \caption{The histograms of the posterior distributions for each parameter (Scenario 2).}
    \label{fig:hist_s2}
\end{figure}

\begin{figure}[htbp]
    \centering
    \begin{subfigure}[t]{0.40\textwidth}
        \centering
        \includegraphics[width=\linewidth]{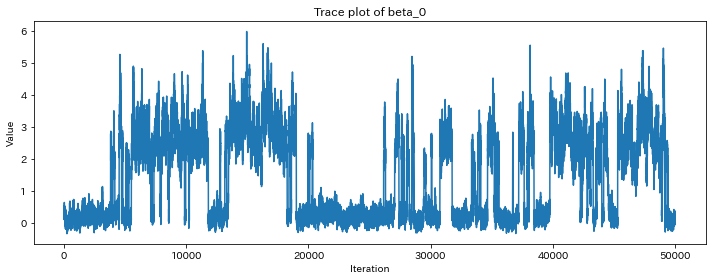}
        \caption{$\beta_0$}
    \end{subfigure}
    \hfill
    \begin{subfigure}[t]{0.40\textwidth}
        \centering
        \includegraphics[width=\linewidth]{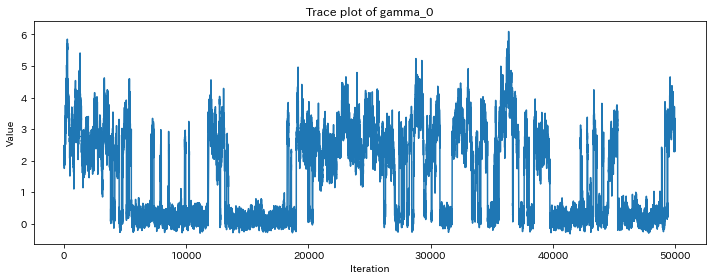}
        \caption{$\gamma_0$}
    \end{subfigure}
    \\
    \begin{subfigure}[t]{0.40\textwidth}
        \centering
        \includegraphics[width=\linewidth]{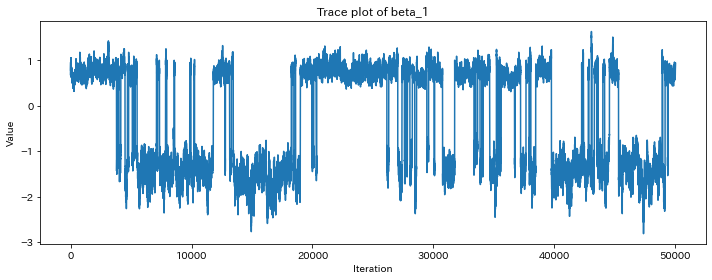}
        \caption{$\beta_1$}
    \end{subfigure}
    \hfill
    \begin{subfigure}[t]{0.40\textwidth}
        \centering
        \includegraphics[width=\linewidth]{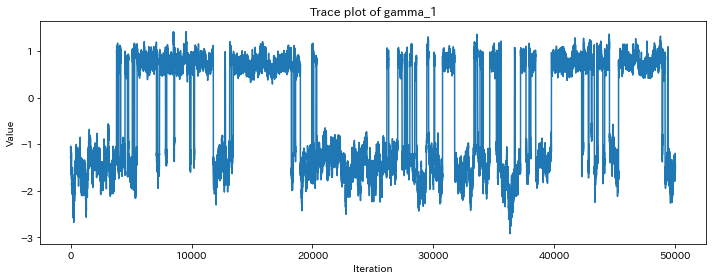}
        \caption{$\gamma_1$}
    \end{subfigure}
    \\
    \begin{subfigure}[t]{0.40\textwidth}
        \centering
        \includegraphics[width=\linewidth]{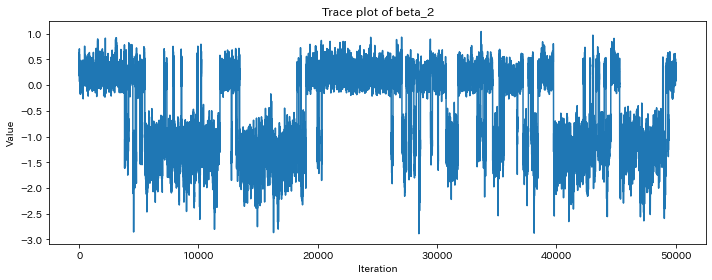}
        \caption{$\beta_2$}
    \end{subfigure}
    \hfill
    \begin{subfigure}[t]{0.40\textwidth}
        \centering
        \includegraphics[width=\linewidth]{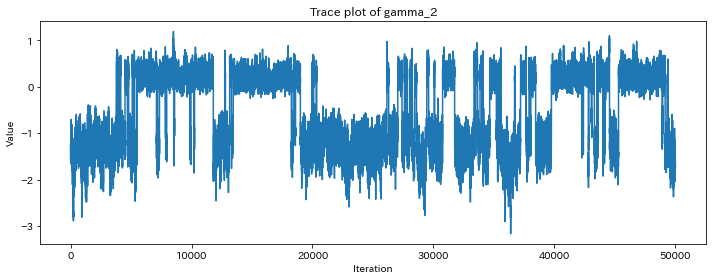}
        \caption{$\gamma_2$}
    \end{subfigure}
    \\
    \begin{subfigure}[t]{0.40\textwidth}
        \centering
        \includegraphics[width=\linewidth]{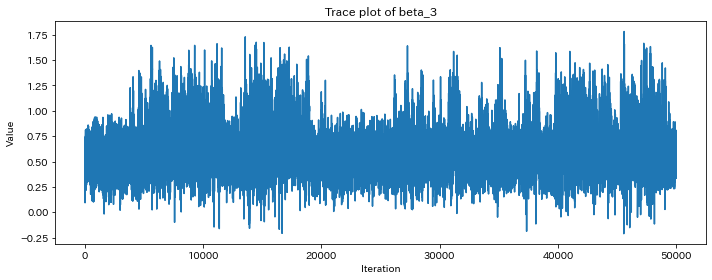}
        \caption{$\beta_3$}
    \end{subfigure}
    \hfill
    \begin{subfigure}[t]{0.40\textwidth}
        \centering
        \includegraphics[width=\linewidth]{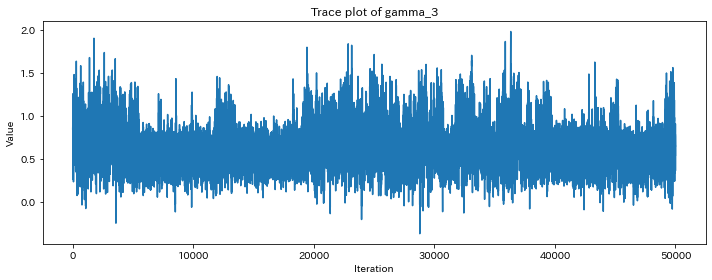}
        \caption{$\gamma_3$}
    \end{subfigure}
    \\
    \begin{subfigure}[t]{0.40\textwidth}
        \centering
        \includegraphics[width=\linewidth]{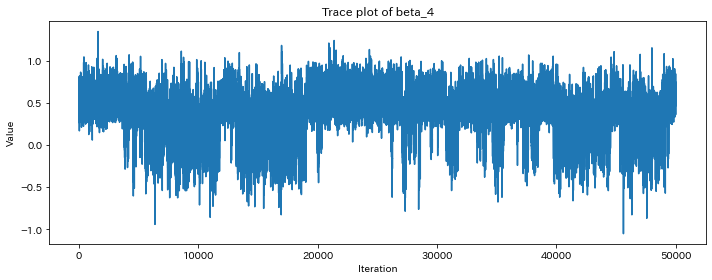}
        \caption{$\beta_4$}
    \end{subfigure}
    \hfill
    \begin{subfigure}[t]{0.40\textwidth}
        \centering
        \includegraphics[width=\linewidth]{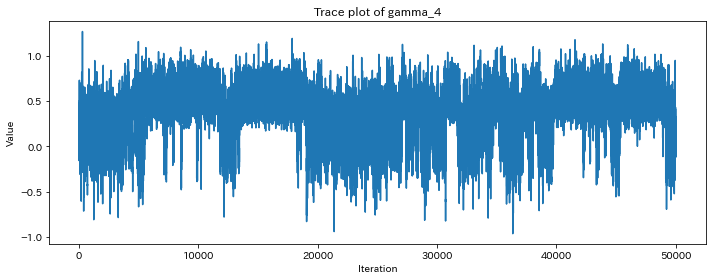}
        \caption{$\gamma_4$}
    \end{subfigure}
    \caption{Trace plots of each parameter (Scenario 3).}
    \label{fig:trace_s3}
\end{figure}

\begin{figure}[htbp]
    \centering
    \begin{subfigure}[t]{0.35\textwidth}
        \centering
        \includegraphics[width=\linewidth]{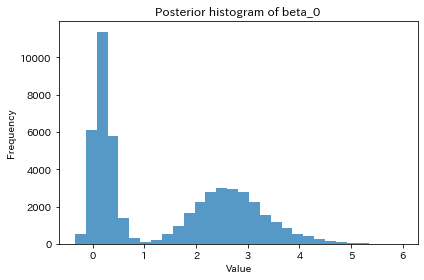}
        \caption{$\beta_0$}
    \end{subfigure}
    \hfill
    \begin{subfigure}[t]{0.35\textwidth}
        \centering
        \includegraphics[width=\linewidth]{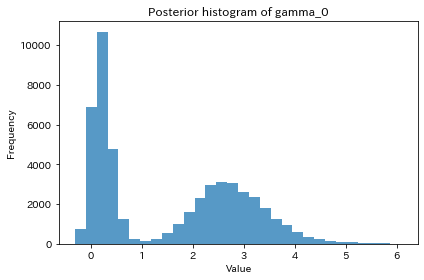}
        \caption{$\gamma_0$}
    \end{subfigure}
    \\
    \begin{subfigure}[t]{0.35\textwidth}
        \centering
        \includegraphics[width=\linewidth]{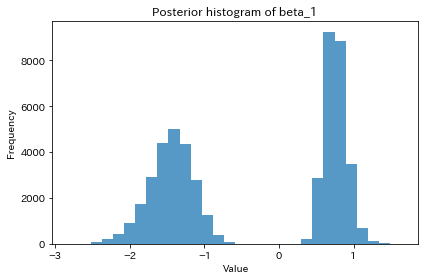}
        \caption{$\beta_1$}
    \end{subfigure}
    \hfill
    \begin{subfigure}[t]{0.35\textwidth}
        \centering
        \includegraphics[width=\linewidth]{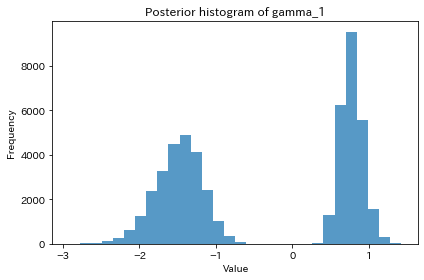}
        \caption{$\gamma_1$}
    \end{subfigure}
    \\
    \begin{subfigure}[t]{0.35\textwidth}
        \centering
        \includegraphics[width=\linewidth]{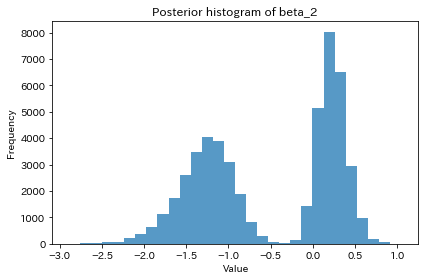}
        \caption{$\beta_2$}
    \end{subfigure}
    \hfill
    \begin{subfigure}[t]{0.35\textwidth}
        \centering
        \includegraphics[width=\linewidth]{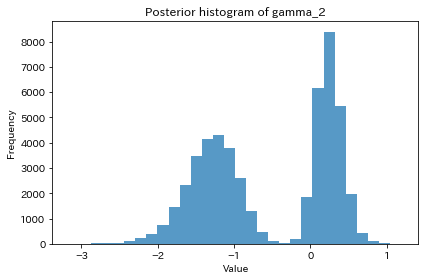}
        \caption{$\gamma_2$}
    \end{subfigure}
    \\
    \begin{subfigure}[t]{0.35\textwidth}
        \centering
        \includegraphics[width=\linewidth]{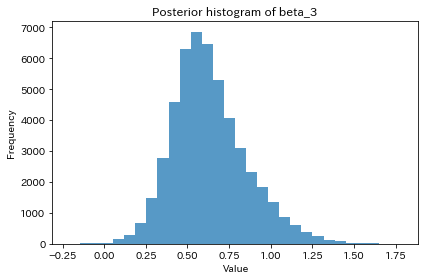}
        \caption{$\beta_3$}
    \end{subfigure}
    \hfill
    \begin{subfigure}[t]{0.35\textwidth}
        \centering
        \includegraphics[width=\linewidth]{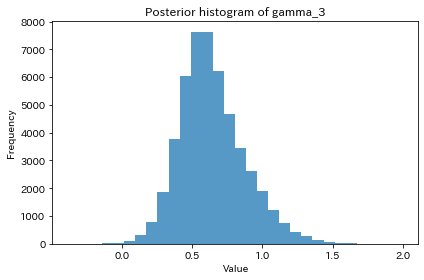}
        \caption{$\gamma_3$}
    \end{subfigure}
    \\
    \begin{subfigure}[t]{0.35\textwidth}
        \centering
        \includegraphics[width=\linewidth]{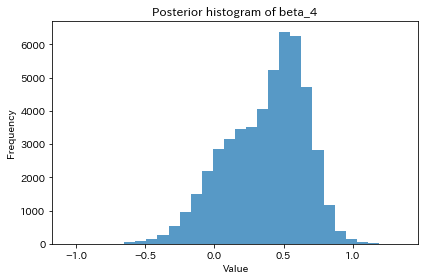}
        \caption{$\beta_4$}
    \end{subfigure}
    \hfill
    \begin{subfigure}[t]{0.35\textwidth}
        \centering
        \includegraphics[width=\linewidth]{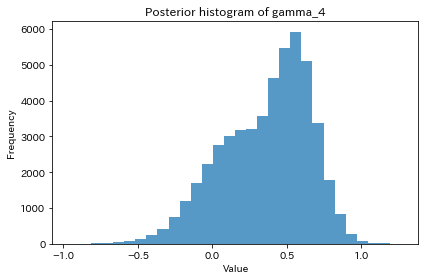}
        \caption{$\gamma_4$}
    \end{subfigure}
    \caption{The histograms of the posterior distributions for each parameter (Scenario 3).}
    \label{fig:hist_s3}
\end{figure}

\bibliography{bibliography}
\bibliographystyle{apalike}

\end{document}